\documentclass{scrartcl}

\usepackage{booktabs}

\usepackage{cite}
\usepackage{amsmath,amssymb,amsfonts}
\usepackage{algorithmic}
\usepackage{graphicx}
\usepackage{textcomp}
\usepackage{xcolor}
\def\BibTeX{{\rm B\kern-.05em{\sc i\kern-.025em b}\kern-.08em
    T\kern-.1667em\lower.7ex\hbox{E}\kern-.125emX}}

\usepackage{amsmath,amssymb,amsthm, mathtools}
\usepackage{booktabs}
\usepackage[vlined,linesnumbered,ruled]{algorithm2e}
\usepackage{bbm}
\usepackage{verbatim}
\usepackage{paralist}
\usepackage{pgfplotstable}
\usepackage{pgfplots}
\usepackage{pifont}
\usepackage{dsfont}
\usepackage{xcolor,graphicx}
\usepackage{caption}
\usepackage{ctable}
\usepackage{authblk}

\usepackage{multirow}
\usepackage{tabularx}
\usepackage[disable]{todonotes} 
\usepackage{scrtime}
\usepackage{eso-pic}

\usepackage[sort,numbers]{natbib}
\usepackage{subfigure}	
\usepackage{tikz}
\usetikzlibrary{arrows.meta}
\usetikzlibrary{decorations.pathreplacing,backgrounds}
\usetikzlibrary{fit,decorations.pathreplacing,arrows,patterns,automata,calc,plotmarks}
\newcolumntype{C}[1]{>{\centering\arraybackslash}p{#1}}
\usepackage[pdfdisplaydoctitle,menucolor=orange!40!black,filecolor=magenta!40!black,urlcolor=blue!40!black,linkcolor=red!40!black,citecolor=green!40!black,colorlinks]{hyperref}
\usepackage[nameinlink, capitalise]{cleveref}
\usepackage{siunitx}
\usepgfplotslibrary{groupplots}
\usepackage{pgfplotstable}

\definecolor{italyGreen}{RGB}{0, 146, 70}
\definecolor{italyRed}{RGB}{206, 43, 55}

\pgfplotsset{
    discard if not/.style 2 args={
        x filter/.code={
            \edef\tempa{\thisrow{#1}}
            \edef\tempb{#2}
            \ifx\tempa\tempb
            \else
                
            \fi
        }
    }
}
\newcommand{\boxPlotWidth}{0.33em}

\pgfplotsset{
    box plot/.style={
        /pgfplots/.cd,
        only marks,
        mark=-,thick,
        mark size=\boxPlotWidth,
        /pgfplots/error bars/.cd,
        y dir=plus,
        y explicit,
    },
    box plot box/.style 2 args={
        /pgfplots/error bars/draw error bar/.code 2 args={%
            \draw[thick]  ##1 -- ++(\boxPlotWidth,0pt) |- ##2 -- ++(-\boxPlotWidth,0pt) |- ##1 -- cycle;
        },
        /pgfplots/table/.cd,
       y={#2},
       y error expr={\thisrow{#1}-\thisrow{#2}},
       /pgfplots/box plot
   },
   box plot top whisker/.style 2 args={
       /pgfplots/error bars/draw error bar/.code 2 args={%
           \pgfkeysgetvalue{/pgfplots/error bars/error mark}%
           {\pgfplotserrorbarsmark}%
           \pgfkeysgetvalue{/pgfplots/error bars/error mark options}%
           {\pgfplotserrorbarsmarkopts}%
            \path[thick] ##1 -- ##2;
      },        /pgfplots/table/.cd,
        y={#1},
        y error expr={\thisrow{#2}-\thisrow{#1}},
        /pgfplots/box plot
    },
   box plot bottom whisker/.style 2 args={
       /pgfplots/error bars/draw error bar/.code 2 args={%
            \pgfkeysgetvalue{/pgfplots/error bars/error mark}%
           {\pgfplotserrorbarsmark}%
          \pgfkeysgetvalue{/pgfplots/error bars/error mark options}%
            {\pgfplotserrorbarsmarkopts}%
            \path[thick] ##1 -- ##2;
        },
        /pgfplots/table/.cd,
        y={#1},
        y error expr={\thisrow{#2}-\thisrow{#1}},
        /pgfplots/box plot
    },
    box plot median/.style 2 args={
        /pgfplots/table/.cd,
        y={#1},
        /pgfplots/box plot
   }
}

\pgfmathdeclarefunction{lg2}{1}{%
    \pgfmathparse{ln(#1)/ln(2)}%
}
\pgfmathdeclarefunction{lg10}{1}{%
    \pgfmathparse{ln(#1)/ln(10)}%
}
\newtheorem{theorem}{Theorem}

\newtheorem{lemma}{Lemma}
\newtheorem{observation}{Observation}
\newtheorem{proposition}{Proposition}

\theoremstyle{definition}
\newtheorem{definition}{Definition}

\newtheorem{example}{Example}
\newtheorem{transformation}{Transformation}
\newcommand{\tempG}{$\mathcal G = (V,E,T, \alpha, \beta)$}
\newcommand{\tempGT}{$\mathcal G = (V,E,T, \alpha, \beta)$}

\newcommand{\tempGTna}{$\mathcal G = (V,E,T, \beta)$}

\DeclareMathOperator{\opt}{opt}
\DeclareMathOperator{\with}{with}

\DeclareMathOperator{\notempty}{ not~empty}
\DeclareMathOperator{\val}{lin}
\DeclareMathOperator{\valp}{lin^T}
\DeclareMathOperator{\ind}{ind}
\DeclareMathOperator{\alp}{A}

\begin{document}
\title{Efficient Computation of Optimal Temporal Walks under Waiting-Time Constraints}
\author{Anne-Sophie Himmel\thanks{Supported by the DFG, projects FPTinP (NI 369/16).}}
\author{Matthias Bentert}
\author{Andr\'e Nichterlein}
\author{Rolf Niedermeier}
\affil{\small Algorithmics and Computational Complexity, Faculty IV, TU~Berlin, Berlin, Germany,

\texttt{\{anne-sophie.himmel, matthias.bentert, andre.nichterlein, rolf.niedermeier\}@tu-berlin.de}}
%
\date{}
\maketitle

\begin{abstract} 
Node connectivity plays a central role in temporal network analysis.
We provide a comprehensive study of various concepts
of walks in temporal graphs, that is, graphs with fixed vertex sets but 
edge sets changing over time. 
Taking into account the temporal aspect leads to a rich set of optimization criteria for ``shortest''
walks. Extending and significantly broadening state-of-the-art work of Wu et al.~[IEEE~TKDE~2016],
we provide an
algorithm for computing optimal walks that is capable to deal with various optimization criteria and 
any linear combination of these. It runs in $O (|V| + |E| \log |E|)$ time where $|V|$ is the number of vertices and $|E|$ is the number of time edges.
A central distinguishing factor to Wu et al.'s work is that 
our model allows to, motivated by real-world applications, respect waiting-time constraints for vertices, 
that is, the minimum and maximum waiting time allowed in intermediate vertices of a walk. 
Moreover, other than Wu et al.\ our algorithm also allows to search for walks that pass multiple subsequent edges in one time step,  
and it can optimize a richer set of optimization criteria. Our experimental studies 
indicate that our richer modeling can be achieved without significantly worsening the running time when compared to Wu et al.'s algorithms.
\end{abstract}

\section{Introduction}
\label{introduction}
Computing shortest paths in networks is arguably among the most important graph 
algorithms, relevant in numerous application contexts and being used as a subroutine 
in a highly diverse set of applications. 
While the case has been studied in static graphs for decades, over the last years 
there has been an intensified interest in studying shortest 
path computations in \emph{temporal graphs}---graphs where the vertex set remains static, but the edge set may change over 
(discrete) time.

Two natural motivating examples for the relevance of path and walk (which can visit a vertex multiple times) computations in temporal 
graphs are as follows. First, Wu et al.~\cite{wu2016temporalpath} discuss applications in 
flight networks where every node represents an airport and each edge is labeled with a flight's departure time.
Clearly, a ``shortest'' path may then relate to a most convenient flight connection between two cities.
Second, understanding the spread of infectious diseases is a major challenge to global health. 
Herein, nodes represent persons and time-labeled edges represent contacts between 
persons where say a virus can be transmitted.  ``Shortest'' path (walk) analysis here may help us 
(among other concepts of connectivity) to find measures against disease spreading \cite[Chapter 17]{New18}.
Notably, in both examples one might need to also take into account issues such as different concepts of 
``shortest''---also called optimal---paths (walks) or waiting times in nodes; this will be an important aspect of our 
modeling. 

Our main reference point is the work of Wu et al.~\cite{wu2016temporalpath} on efficient algorithms 
for optimal \emph{temporal path} computation. These are also implemented in the temporal graph library of Apache 
Flink~\cite{Lightenberg:2018:TTG:3184558.3186934}. 
We extend their model with respect to two aspects.
First, we additionally consider waiting-time constraints\footnote{Waiting-time constraints play a particularly important role in standard spreading models of infectious diseases such as the SIS-model (Susceptible-Infected-Susceptible)~\cite[Chapter 17]{New18}.} 
for the network nodes; 
importantly, maximum-waiting-time constraints can enforce cycles from one node to another; 
refer to \autoref{figure:beta} for a simple example. 
Hence, we need to take into account optimal \emph{temporal walks} from one node to another
(in Wu et al.'s model without waiting-time constraints 
there is always an (optimal) temporal path visiting no node twice because no cycles are necessary).
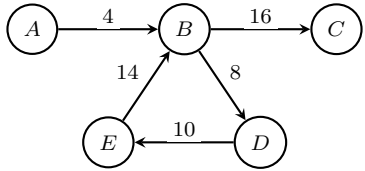
\begin{figure}
\centering
\begin{tikzpicture}
\begin{scope}[every node/.style={circle,thick}]
    \node (A) at (0,0) [draw=black] {\scriptsize $A$};
    \node(B) at (2,-0)[draw=black] {\scriptsize $B$};
    \node (E) at (4,0) [draw=black] {\scriptsize $C$};
    \node (C) at (3,-1.5)[draw=black] {\scriptsize $D$};
    \node (D) at (1,-1.5)[draw=black] {\scriptsize $E$};
\end{scope}
\begin{scope}[>={Stealth[black]},
              every node/.style={ fill=white, above, inner sep=2pt, align=left},
              every edge/.style={draw=black, thick}]
    \path [->,>=stealth] (A) edge node {\scriptsize $4$} (B);
    \path [->,>=stealth] (B) edge node[xshift=5pt] {\scriptsize $8$} (C);
    \path [->,>=stealth] (C) edge node {\scriptsize $10$} (D);
    \path [->,>=stealth] (D) edge node[xshift=-7pt] {\scriptsize $14$} (B);
    \path [->,>=stealth] (B) edge node {\scriptsize $16$} (E);
\end{scope}
\end{tikzpicture}
	\caption{A temporal graph (with time-labeled arcs) 
	with maximum waiting time four in which the only temporal walk from $A$ to $C$ visits $B$ twice.}
\label{figure:beta}
\end{figure}
Actually, if one insists on paths (without repeated nodes) instead of walks, then even deciding whether there exists a path between two nodes becomes NP-hard \cite{restlessTempPathArxiv19}.
The second extension to Wu et al.'s work lies in an increased number of optimality criteria (different notions of optimal walks) and the fact that we do not only deal with optimizing one criterion but a linear combination of any of these, thus addressing richer modeling needs in real-world applications.
For example, in the above-mentioned flight network a traveling person might want to use our algorithm to optimize a linear combination of cost and travel duration. 
Note that trying to find walks under constraints (e.g., travel duration at most~$t$ \emph{and} cost at most~$c$ for given~$c,t \in \mathds N$) leads to NP-hard computational problems, even in the static case~\cite{AMO93}.


\paragraph{Related Work.} 
The theory of temporal graphs is a relatively young but active field of research \cite{Holme2015,holme2012temporal,holme2013temporal,holme2019temporal}.
Many of the basic concepts of temporal graphs such as temporal connectivity~\cite{kempe2000conn,nicosia2012components,axiotis2016size,mer13} or $s$-$z$-separation~\cite{ZFMN20,FLUSCHNIK2019} are based on the notion of the temporal paths and walks.
 The concept of optimal temporal walks plays a central role in the definition of temporal graph metrics such as eccentricity, diameter, betweenness and closeness centrality~\cite{Saramaki2011,santoro2011time,kim2012}.
 %
 %
%

An early algorithm for computing optimal temporal walks is due to Xuan et al.~\cite{xuan2003computing}. 
They computed temporal walks under different optimization criteria, namely \textit{foremost}, \textit{fastest}, and \textit{minimum hop-count}\footnote{Refer to the next section for definitions of these and further optimality criteria.} for a restricted variant of temporal graphs.  
Wu et al.~\cite{wu2016temporalpath} followed up by introducing algorithms for computing optimal walks for the 
optimization criteria \textit{foremost}, \textit{reverse-foremost}, \textit{fastest}, and~\textit{shortest} on temporal graphs with no waiting-time constraints. Their algorithms run in linear and quasi-linear time with respect to the number of time-arcs, provided that transmission times on time-arcs are greater than zero. 
Concerning time-dependent multicriteria optimal path computation, there has been research in the related field of route planning~\cite{bast2016route}.

The study of minimum- and maximum-waiting-time constraints in vertices has not received much attention in the context of temporal walks even though they are considered as important extensions to the temporal walk model \cite{holme2012temporal,Saramaki2011}. 
One of the first papers to consider waiting-time constraints in the context of paths is the one by Dean~\cite{Dean2004AlgorithmsFM}. He showed that time-dependent shortest path problems with waiting constraints can be solved in polynomial time which is closely related to waiting-time constraints for optimal temporal walks. 
In more recent work, Modiri et al.~\cite{modiri2019efficient} and Kivel\"a et al.~\cite{kivela2018mapping} studied the changes in reachability when introducing maximum-waiting time constraints for temporal walks using so-called \emph{event graphs}. 
Lastly, Casteigts et al.~\cite{restlessTempPathArxiv19} have recently shown that finding optimal temporal path under maximum-waiting-time constraints is NP-hard.

In the so-called multistage setting, which is closely related to temporal graphs, paths have been studied by Fluschnik et al.~\cite{fluschnikundzschoche}.

\paragraph{Our Contributions.}
We analyze the running time complexity of computing optimal temporal walks under waiting-time constraints. 
We develop and (theoretically and empirically) analyze an algorithm for finding an optimal walk from a 
source vertex to each vertex in the temporal graph under waiting-time constraints. 
Our algorithm runs in quasi-linear time in the number of time-arcs.  
This implies that the introduction of waiting-time constraints on temporal walks does not 
increase the asymptotic computational complexity of finding optimal temporal walks.   
Moreover, our algorithm can compute optimal walks not only for single optimality criteria 
but also for any linear combination of these.  
In experiments on real-world social network data sets, we demonstrate that in terms of 
efficiency our algorithm can compete with 
state-of-the-art algorithms by Wu et al.~\cite{wu2016temporalpath}. Their algorithms only run on temporal graphs 
without waiting-time constraints, which do optimize only one criterion (and not a linear combination).
Additionally, our algorithm allows transmission times being zero (hence, allowing to pass multiple arcs in one time step) while the algorithms of Wu et al.~request transmission times on arcs to be greater than zero.
Lastly, we analyze how different values of maximum waiting times influence the reachability within the temporal graph and how they influence the structure of optimal temporal walks in context of the existence of cycles.  
\paragraph{Organization of the Paper.}
In \cref{sec:modeling}, we discuss temporal graphs, temporal walks, and the 
various corresponding  optimality criteria, 
starting with an extensive motivating example in the context of disease spreading.
In \cref{sec:prelim}, we introduce definitions and notations used throughout the paper. 
We continue in \cref{sec:trans} by presenting two simple linear-time transformations to eliminate 
transmission times and minimum waiting times in the temporal graph 
without loosing any modeling power.
In \cref{sec:algorithm}, we design and analyze an algorithm for computing optimal walks 
under maximum-waiting-time for any linear combination of these. 
Finally, in \cref{sec:exp}, we demonstrate the efficiency of our algorithm on real-world data sets. 
We compare our running times with the running times of the algorithms of Wu et al.~\cite{wu2016temporalpath}. 
We further examine the impact of different maximum-waiting-time values on the existence and structure of optimal temporal walks.

\section{Modeling of Optimal Temporal Walks}
\label{sec:modeling}
Before we introduce our basic concepts relating to temporal graphs and walks, we start 
with a more extensive discussion of a motivating example from the disease spreading context.

\paragraph*{Disease Spreading Motivating Example.}
Pandemic spread of an infectious disease is a great threat to global health, potentially associated with high mortality rates as well as economic fallout~\cite{Salath22020}.
Understanding the dynamics of infectious disease spread within human proximity networks could facilitate the development of mitigation strategies. 

A large part of the legwork required to understand the dynamics of infectious diseases is the analysis of transmission routes through proximity networks~\cite{Salath22020}.  
Classical graph theory can be used to model the main structure of a network: 
Each person in the network is represented by a node and an bi-directional arc between two nodes indicates at least one proximity contact between these persons. 
However, the time component plays a crucial role in the analysis of transmission routes of a potential disease as shown in the following example:
\begin{example}
Studying a proximity network as shown in~\cref{figure:motiv1}, there are several transmission routes from~$A$ to~$D$, e.g.,\ $A\rightarrow B \rightarrow D$ and $A\rightarrow C \rightarrow D$, by which a disease could have spread. 
If we extend our model by the points in time of proximity contacts in~\cref{figure:motiv2} to \cref{figure:motiv4}, then we reach the conclusion that a disease could not have spread from~$A$ to~$D$. 
The proximity contacts~$A\overset{3}{\rightarrow} B$ and~$A \overset{3}{\rightarrow} C$ occurred on day three whereas the contacts~$B \overset{1}{\rightarrow} D$ and~$C \overset{2}{\rightarrow} D$ occurred on days one and two, respectively. 
Thus, $A$ could have only infected $B$ and $C$ after proximity contact with~$D$. 
\label{example:1}
\end{example}
In addition to what has been said so far, the infectious period of a disease also has to be taken into account when computing potential transmission routes through the network, implying the minimum time a person has to be infected before she becomes contagious herself and the maximum time a person can be infected before she is no longer contagious:

\begin{example}
If person~$B$ was infected by person~$A$ on day four ($A \overset{4}{\rightarrow} B$) and the infectious period of the disease starts after one day and ends after the fourth day, then person~$B$ could not have infected person~$C$ she met on day ten ($B \overset{10}{\rightarrow} C$). 
Hence, person~$C$ could not have been infected by the disease via the transmission route $A \overset{4}{\rightarrow} B \overset{10}{\rightarrow}C$.  
\label{example:2}
\end{example}

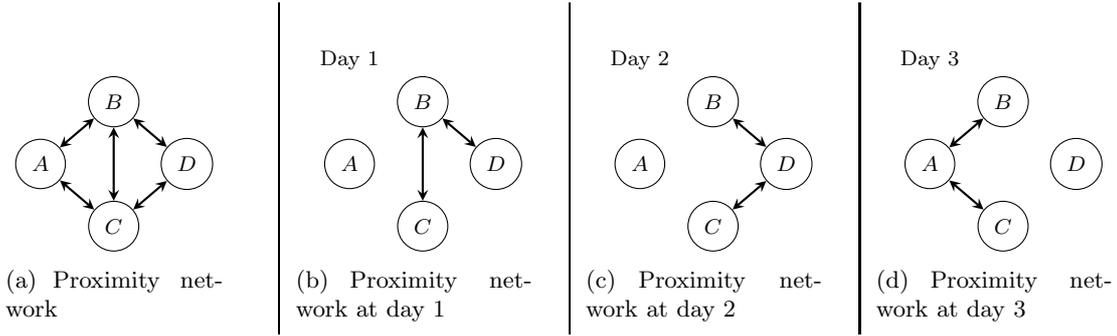
\begin{figure}[t]
\begin{tabular}{ b{.23\textwidth} | b{.23\textwidth}  |b{.23\textwidth}  | b{.23\textwidth} }
{\centering
\subfigure[Proximity~network]{
\begin{tikzpicture}[scale=.55]
\begin{scope}[every node/.style={circle}]
    \node (A) at (-1.75,0)[draw=black] {\scriptsize $A$};
    \node (B) at (0,1.5)[draw=black] {\scriptsize $B$};
    \node (C) at (0,-1.5)[draw=black] {\scriptsize $C$};
    \node (D) at (1.75,0)[draw=black] {\scriptsize $D$};
\end{scope}
\begin{scope}[>={Stealth[black]},
              every node/.style={ fill=white,above, inner sep=2pt, align=left},
              every edge/.style={draw=black, thick}]
    \path [<->,>=stealth] (A) edge (B);
    \path [<->,>=stealth] (A) edge (C);
    \path [<->,>=stealth] (B) edge (C);
    \path [<->,>=stealth] (B) edge (D);
    \path [<->,>=stealth] (C) edge (D);
\end{scope}
\label{figure:motiv1}
\end{tikzpicture}
}} &
{\centering
\subfigure[Proximity network at day $1$]{
\begin{tikzpicture}[scale=.55]
\begin{scope}[every node/.style={circle}]
    \node (t) at (-1.75,2.5)[draw=white] {\scriptsize{Day $1$}};
    \node (A) at (-1.75,0)[draw=black] {\scriptsize $A$};
    \node (B) at (0,1.5)[draw=black] {\scriptsize $B$};
    \node (C) at (0,-1.5)[draw=black] {\scriptsize $C$};
    \node (D) at (1.75,0)[draw=black] {\scriptsize $D$};
\end{scope}
\begin{scope}[>={Stealth[black]},
              every node/.style={ fill=white,above, inner sep=2pt, align=left},
              every edge/.style={draw=black, thick}]
    \path [ <->,>=stealth] (B) edge  (C);
    \path [<->,>=stealth] (B) edge  (D);

\end{scope}
\label{figure:motiv2}
\end{tikzpicture}
} } &
{\centering
\subfigure[Proximity network at day $2$]{
\begin{tikzpicture}[scale=.55]
\begin{scope}[every node/.style={circle}]
    \node (t) at (-1.75,2.5)[draw=white] {\scriptsize{Day $2$}};
    \node (A) at (-1.75,0)[draw=black] {\scriptsize $A$};
    \node (B) at (0,1.5)[draw=black] {\scriptsize $B$};
    \node (C) at (0,-1.5)[draw=black] {\scriptsize $C$};
    \node (D) at (1.75,0)[draw=black] {\scriptsize $D$};
\end{scope}
\begin{scope}[>={Stealth[black]},
              every node/.style={ fill=white,above, inner sep=2pt, align=left},
              every edge/.style={draw=black, thick}]
    \path [<->,>=stealth] (B) edge (D);
    \path [<->,>=stealth] (C) edge (D);
\end{scope}
\label{figure:motiv3}
\end{tikzpicture}
}} &
{\centering
\subfigure[Proximity network at day $3$]{
\begin{tikzpicture}[scale=.55]
\begin{scope}[every node/.style={circle}]
    \node (t) at (-1.75,2.5)[draw=white] { \scriptsize{Day $3$}};
    \node (A) at (-1.75,0)[draw=black] {\scriptsize $A$};
    \node (B) at (0,1.5)[draw=black] {\scriptsize $B$};
    \node (C) at (0,-1.5)[draw=black] {\scriptsize $C$};
    \node (D) at (1.75,0)[draw=black] {\scriptsize $D$};
\end{scope}
\begin{scope}[>={Stealth[black]},
              every node/.style={ fill=white,above, inner sep=2pt, align=left},
              every edge/.style={draw=black, thick}]
    \path [<->,>=stealth] (A) edge (B);
    \path [<->,>=stealth] (A) edge (C);
\end{scope}
\label{figure:motiv4}
\end{tikzpicture}
}}
\end{tabular}
\caption{A proximity network modeled as a static graph~(\cref{figure:motiv1}) and a closer look at the days in which the proximity contacts appear~(\cref{figure:motiv2,figure:motiv3,figure:motiv4}).}
\label{figure:motiv}
\end{figure}
\paragraph{Temporal Graphs.} 
These are capable of representing both properties elaborated in the two examples above. Temporal graphs are already a frequently used model in the prediction and control of infectious diseases~\cite{masuda2013predicting, holme16}.
Temporal graphs---also referred to as temporal networks~\cite{nicosia2013graph, holme2012temporal}, evolving graphs~\cite{xuan2003computing}, or time-varying graphs~\cite{santoro2011time,casteigts2012time}---are graphs where the arc set changes over time\footnote{We always assume that all changes over time are given as input.}; thus, they can capture the dynamics within a proximity network. 

In this paper, we will consider the following temporal graph model: A temporal graph consists of a \emph{lifetime}, a \emph{set of vertices}, and a \emph{set of time-arcs}. A time-arc is a directed edge between two vertices that is associated with a \emph{time stamp} at which the contact occurs and a \emph{transmission time} that indicates the amount of time to traverse the arc. Furthermore, each vertex~$v$ exhibits an individual \emph{minimum waiting time} $\alpha(v)$ and \emph{maximum waiting time} $\beta(v)$ that can reflect the infectious period in our previous example.  

The application areas of temporal graphs are numerous: In addition to human and animal proximity networks, they are used in communication networks, traffic networks, and distributed computing, to name a few application areas~\cite{holme2012temporal,holme2013temporal,holme2019temporal}.

In our running disease spreading example, we are interested in transmission routes of an infectious disease. 
These transmission routes can revisit a person in the proximity network due to possible reinfection \cite{Bar16}. 
Hence, the transmission routes can contain cycles which needs to be considered in the choice of concepts representing these routes.      
\paragraph{Temporal Walks \& Optimal Temporal Walks.} 
Within the temporal graph model, temporal walks---also called journeys \cite{xuan2003computing,nicosia2013graph}---are the fundamental concept that represents the transmission routes in our running example.  

 A temporal walk is a sequence of time-arcs which connects a sequence of vertices and which are non-decreasing in time. In our model, a temporal walk additionally ensures that it remains the minimum waiting time in each intermediate vertex and does not exceed the maximum waiting time in any intermediate vertex of the walk. 
\begin{example}
Continuing \cref{example:2}, a valid temporal walk (transmission route) from $A$ to $D$ could be the following: $A \overset{4}{\rightarrow} B \overset{8}{\rightarrow} D$. Person $A$ could have infected $B$ on day four. 
Due to the infectious period of four days, $B$ was still contagious on day eight when she had contact with person $C$. 
This does not hold on a route~$A \overset{4}{\rightarrow} B \overset{10}{\rightarrow} C$ as discussed in \cref{example:2}. 
If $B$ was infected at time step $4$,  then she was not contagious anymore at time step $10$. 
\end{example}

A \emph{temporal path} is a temporal walk where all vertices are pairwise distinct. 
Maximum-waiting-time constraints have significant impact on temporal walks.
In a temporal graph with constraints on the maximum waiting time, one can be forced to make detours because the maximum waiting time in a vertex is exceeded.
As a consequence, there can be two vertices $A$ and~$C$ such that any temporal walk from $A$ to~$C$ is not a path, as shown in \autoref{figure:beta}.
\begin{observation}
Let \tempGTna~be a temporal graph with maximum-waiting-time constraints.
Then there can exist two vertices~$s,z\in V$ such that each $s$-$z$-walk is not a temporal path.
 \label{observation:beta}
\end{observation}

We are interested in temporal walks within our proximity network in general, but wish to place emphasis on temporal walks that optimize certain properties. 
A plethora of criteria can be optimized as a consequence of the time aspect. 
Possible criteria (with the names we chose or were chosen in literature in brackets) include: arrival time~(\textit{foremost}), departure time~(\textit{reverse-foremost}), duration~(\textit{fastest}), transmission time~(\textit{shortest}), number of time-arcs~(\textit{minimum hop-count}), time-arc cost~(\textit{cheapest}), 
probability~(\textit{Most-likely}), and waiting time~(\textit{minimum waiting time}). 
Next, we provide examples for all properties from their respective fields of application.
%
\begin{description}
\item[\textit{Foremost.}] A foremost walk is a temporal walk that has the earliest arrival time possible. 
Computing a foremost walk from a source vertex to all vertices in the proximity network signifies the speed with which an infectious disease could spread. 
 
\item[\textit{Reverse-Foremost.}] A reverse-foremost walk is a temporal walk that exhibits the latest possible departure time. 
Computing a reverse-foremost walk from a source vertex to all vertices in the proximity network estimates the latest possible point in time at which an infectious disease could start spreading and still permeates the entire network. 

\item[\textit{Fastest.}] A fastest walk is a temporal walk which exhibits the minimum duration, that is, the minimum difference between departure and arrival times. 
For an appropriate motivation, we leave proximity networks and consider the field of flight networks. 
Airports represent vertices, time-arcs represent flights from one airport to another. 
The time stamp indicates the departure time of a flight, the transmission time indicates the duration. 
The minimum waiting time in the vertices signifies the minimum time required in an airport to catch a connecting flight. 
Within flight networks, the duration is often the criterion passengers aim to minimize in order to streamline their journey. 
    
\item[\textit{Shortest.}] A shortest walk is a temporal walk that minimizes the sum of transmission times on the time-arcs. 
In the context of flight networks, a shortest walk is a flight connection with the minimum time spent airborne. 

\item[\textit{Minimum Hop-Count.}] A minimum-hop-count walk is a temporal walk which minimizes the number of time-arcs. 
Within a flight network, passengers also aim to minimize their number of connecting flights to avoid lengthy boarding procedures and the risk of missing connecting flights.

\item[\textit{Cheapest.}] For a given cost function on the time-arcs, a cheapest walk is a temporal walk with the minimum sum of costs over all time-arcs. 
The benefits of the minimization of this property within flight networks are obvious: Weighing long travel times and multiple connections against the cheapest fare is the oldest consideration in the book for many air travelers.  
 
\item[\textit{Most-Likely.}] For given probabilities on the time-arcs, a most-likely walk is a temporal walk with the highest probability. 
One application lies in disease spreading: For every contact there is a certain likelihood for an infectious disease to be transmitted depending on the proximity of the persons or the body contact between them. Thus, a most-likely walk is a transmission route with the highest probability for the infectious disease to be spread under the assumption of stochastic independence. The respective probabilities of the time-arcs within the walk are multiplied.    

\item[\textit{Minimum Waiting Time.}] 
The minimum-waiting-time walk is a temporal walk that has minimum sum of waiting times over all intermediate vertices. 
Routing packets through a router network prioritizes minimum waiting times of packages in the routers to improve the overall performance of the network.  
\end{description}


\section{Formal Definitions}
\label{sec:prelim}
In this section, we formally introduce the most important concepts related to temporal graphs, temporal walks, and formalize our optimality criteria. 
We start with some basic mathematical definitions.
We refer to an interval~$[a,b]$ as a contiguous ordered set of discrete time steps:
$[a,b]:=\{n \mid n \in \mathds N \wedge a \leq n \leq b \},$ where~$a,b \in \mathds N$. Further, let~$[a]:=[1,a]$. Given a function $f \colon A \rightarrow B$, we write $f \equiv c$ if  $f(a) = c$ for all $a \in A$ and $c \in B$.

\paragraph{Temporal Graph.}
A temporal graph is a graph whose edge set changes over time. 
\begin{definition}[Temporal Graph]
 A \emph{temporal graph}~$\mathcal G = (V,E,T, \alpha, \beta)$ is a five-tuple consisting of 
 \begin{itemize}
 \item a lifetime $T \in \mathds N$, 
 \item a vertex set~$V$, 
 \item a time-arc set~$E \subseteq V \times V \times \{1,\ldots,T\} \times \{0,\ldots,T\}$, 
 \item a minimum waiting time $\alpha \colon V \rightarrow \{0,\ldots,T\}$, and 
 \item a maximum waiting time $\beta \colon V \rightarrow \{0,\ldots,T\}$.
 \end{itemize}
 \end{definition}
A time-arc~$(v,w, t, \lambda) \in E$ is a directed connection from~$v$ to~$w$ with \emph{time stamp}~$t$ and \emph{transmission time}~$\lambda$, that is, a transmission from~$v$ to~$w$ starting at time step~$t$ and taking~$\lambda$ time steps to cross the arc. The \emph{departure time} in vertex~$v$ is~$t$; the \emph{arrival time} in vertex~$w$ is then~$t+\lambda$. 
%
The two waiting-time functions $\alpha\colon V \rightarrow \mathds N$ and~$\beta \colon V \rightarrow \mathds N$ assign each vertex a minimum and maximum waiting time, respectively. The minimum waiting time~$\alpha(v)$ is the minimum time a person has to stay in a vertex~$v$ before she can move on in the temporal graph. The maximum waiting time~$\beta(v)$ is the maximum time a person can stay in a vertex $v$ before she is no longer allowed to move further in the graph.
A temporal graph~\tempGT~is called \emph{instantaneous} if $\alpha(v) = 0$ for all $v \in V$ and $\lambda=0$ for all~$(v,w,t,\lambda) \in E$. Then, for the ease of presentation, we neglect $\alpha$ and we write arcs as triples $(v,w,t) \in E$ for instantaneous graphs.
\begin{table}
\caption{Frequently used notation for a temporal graph $\mathcal G$.}
\begin{tabularx}{\columnwidth}{ c l }
	\toprule
    $V$ & the vertex set of~$\mathcal G$  \\ 
    $E$ &  the time-arc set of~$\mathcal G$ \\ 
    $\lbrack T \rbrack$ & the time interval of~$\mathcal G$ \\
    $\alpha$  &  the minimum waiting time with $\alpha: V \rightarrow \mathds N$; \\
    $\beta$  &  the maximum waiting time with $\beta: V \rightarrow \mathds N$ \\
    $V_t$  & the vertex subset $V_t \subseteq V$ at time $t$, that is, \\
    & $V_t := \{v  \mid (v,w,t,\lambda) \in E \vee (w,v,t,\lambda) \in E \}$ \\
    $E_t$  &  the time-arc subset at time $t$, that is, $E_t := \{(v,w) \mid (v,w,t,\lambda) \in E \}$ \\
    $G_t$  & the directed, static graph $G_t := (V_t,E_t)$ \\
	\bottomrule
\end{tabularx}
\label{table:notations}
\end{table}
In \cref{table:notations}, we introduce some notation for temporal graphs.
\paragraph{Temporal Walk.} 
A temporal walk is a walk in a temporal graph such that the time stamps of the visited time-arcs of a temporal walk are increasing in time. Additionally, the transmission time and the waiting-time constraints have to be taken into account. 
\begin{definition}[Temporal Walk]
Given a temporal graph \tempG~and two vertices $s,z \in V$, a \emph{temporal walk} from~$s$ to $z$ is a sequence $\left(\left (v_{i-1},v_i,t_i,\lambda_i \right)\right)_{i=1}^k$ of time-arcs such that $s= v_0$, $z = v_k$, and $t_i + \lambda_i + \alpha(v_{i}) \leq t_{i+1} \leq t_i + \lambda_i + \beta(v_{i})$ for all $i\in [k-1]$.
\end{definition}
A \emph{temporal path} is a temporal walk where all vertices are pairwise distinct. 
\paragraph{Optimal Temporal Walk.}
Due to the additional time aspect, there are several, potentially contradicting criteria that can be optimized in a temporal walk. 
We formally define the criteria that were already motivated in \cref{sec:modeling}. 

\begin{definition}[Optimal Temporal Walk]
Let \tempG~be a temporal graph, let~$c\colon E \rightarrow \mathds N$ be a cost function, and let $s,z \in V$ be two vertices. A temporal walk~$P=\left(\left (v_{i-1},v_i,t_i,\lambda_i \right)\right)_{i=1}^k$ from $s$ to $z$ is called \emph{optimal} if it minimizes or maximizes a certain value among all temporal walks from~$s$ to~$z$. We consider the following variants:
\begin{center}
	\renewcommand{\arraystretch}{1.2}
	\begin{tabularx}{\linewidth}{ r c l}
		\toprule
		criterion & min / max & optimization value \\
		\midrule
		foremost	& min	& $t_k + \lambda_k$		\\
		reverse-foremost	& max	& $t_1$		\\
		fastest 	& min	& $(t_k + \lambda_k) - t_1$	\\ 
		shortest 	& min	& $\sum_{i=1}^k \lambda_i$	\\ 
		cheapest	& min	& $\sum_{i=1}^k c((v_{i-1},v_i,t_i,\lambda_i))$ \\
		most-likely 	& max	&   $\prod_{i=1}^k c((v_{i-1},v_i,t_i,\lambda_i)) $	\\
		minimum hop-count  & min	& $k$	\\ 
		minimum waiting time		& min 	& $\sum_{i=1}^{k-1} t_{i+1} - (t_i + \lambda_i)$ \\
		\bottomrule
	\end{tabularx}
\end{center}
\end{definition}

Note that the \textit{most-likely} criterion can easily be transformed into \textit{cheapest}. 
For the \textit{most-likely} criterion, the cost values of the time-arcs represent probabilities, implying~$c(e) \in [0,1]$ for all~$e\in E$. 
Hence, maximizing $\prod_{i=1}^k c((v_{i-1},v_i,t_i,\lambda_i) $ is equivalent to minimizing $$\sum_{i=1}^k  -\log c((v_{i-1},v_i,t_i,\lambda_i))$$ of a temporal walk. 
Hence, we neglect considering the \textit{most-likely} criterion separately. 
%
We further call a temporal walk~$P=\left(\left (v_{i-1},v_i,t_i,\lambda_i \right)\right)_{i=1}^k$ from $s$ to $z$ an optimal temporal walk with respect to a \emph{linear combination} with~$\delta_1,\ldots,\delta_7 \in \mathds Q^+_0$ if it minimizes 
\begin{align*}
		\val(P) = ~&\delta_1 \cdot(t_k + \lambda_k) &\textit{Foremost}\\
	+ ~&\delta_2\cdot(-t_1) &\textit{Reverse-Foremost}\\
	+ ~&\delta_3\cdot(t_k  + \lambda_k - t_1) &\textit{Fastest} \\
	+ ~&\delta_4\cdot \sum_{i=1}^k \lambda_i &\textit{Shortest} \\
	+ ~ &\delta_5\cdot \sum_{i=1}^k c((v_{i-1},v_i,t_i,\lambda_i )) &\textit{Cheapest}\\
	+ ~&\delta_6\cdot k &\textit{Minimum Hop-Count}\\
	+ ~&\delta_7\cdot \sum_{i=1}^{k-1} (t_{i+1} - (t_i + \lambda_i))  &\textit{Minimum Waiting Time}  
\end{align*}
among all temporal walks from~$s$ to~$z$.

\section{Transformations}
\label{sec:trans}
To simplify the presentation of the forthcoming algorithm in \cref{sec:algorithm} for computing optimal temporal walks, we design it to run only on \emph{instantaneous} temporal graphs, that is, temporal graphs with no transmission times ($\lambda =0$ for all $(v,w,t,\lambda) \in E$) and no minimum-waiting-time constraints~($\alpha(v) = 0$ for all $v\in V$). 
This is no restriction since we can eliminate these with the following transformation.

\begin{transformation}[Remove $\alpha$ and $\lambda$.] 
 Let~\tempG\ be a temporal graph and let $c\colon E \rightarrow \mathds N$ be a cost function. Transform $(\mathcal G,c)$ into $(\mathcal G',c',c_\lambda,\ind,A)$ where 	
\begin{itemize}
 \item $\mathcal G'$ is an instantaneous temporal graph $\mathcal G'=(V',E',T,\beta')$ with
 \begin{itemize}
 	\item $V' =V \cup V^E$ with $V^E:=\{v_e \mid e \in E \}$,
 	\item $E' = E^O \cup E^I$ with $E^O := \{(v,v_e,t) \mid e=(v,u,t,\lambda) \in E \}$ and	
 								$E^I = \{ (v_e,u,t + \lambda + \alpha(u))  \mid e=(v,u,t,\lambda) \in E \}$, 
 
  	\item $\beta' \colon V' \rightarrow \mathds N$ with $\beta'(v) := \beta(v)$ for all $v \in V$, else $\beta'(v):= T$,
 \end{itemize}
 \item $c' \colon E' \rightarrow \mathds N$ is cost function with $c'(e) := c(\hat e)$ for all $e=(v,v_{\hat e},t) \in E^O $, \\else~${c'(e) :=0}$. 
 \item $c_{\lambda} \colon E' \rightarrow \mathds N$ is a transmission-cost function with $c_{\lambda}(e) := \ell$ for $e=(v_{\hat e},w,t + \ell + \alpha(w)) \in E^I$ and $\hat e = (v,w,t,\ell) \in E$, else ${c_{\lambda}(e) := 0}$.
 \item $\ind \colon V' \rightarrow \{0,1\}$ is a vertex-index function with $\ind(v) := 1$ if $v \in V$, \\else~$\ind(v):=0$.
 \item $\alp(v) \colon V' \rightarrow \{0,1\}$ is an auxiliary function with $A(v) := \alpha(v)$ if $v \in V$, \\else~$A(v):=0$.
\end{itemize}
 \label{transformation}
 \end{transformation}
 
 We now show that any temporal graph can be transformed by \cref{transformation} into an equivalent instantaneous temporal graph in linear time such that 
 any optimal temporal walk in the instantaneous temporal graphs directly corresponds to an optimal temporal walk in the original graph and vice versa.
 Therefore, we have to slightly adapt the formula for the linear combination as shown in the following proposition.
\begin{proposition}
\label{prop:transformation}
%
Let \tempGT~be a temporal graph, let~$c\colon E \rightarrow \mathds N$ be a cost function, and let~$s,z \in V$. 
Let further $(\mathcal G' = (V',E',T', \beta'),c',c_{\lambda},\ind,A )$ be the result of applying \cref{transformation} to $(\mathcal G,c)$. 

For $\delta_1,\ldots,\delta_7\in \mathds Q^+_0$,
there exists a 
temporal walk~$P=\left(\left (v_{i-1},v_i,t_i,\lambda_i \right)\right)_{i=1}^k =\left(e_i\right)_{i=1}^k$ 
from $s$ to $z$ in $\mathcal G$ is optimal with respect to a linear combination of $\delta_1,\ldots,\delta_7$ 
if and only if 
the temporal walk
$$P'= 
\left(\left (v_{i-1},v_{e_i},t_i),(v_{e_i},v_i,t_i+\lambda_i+\alpha(v_i))\right)\right)_{i=1}^k = (e'_i)_{i=1}^{2k} = (v'_{i-1},v'_i,t'_i)_{i=1}^{2k}$$ 
from $s$ to $z$ in $\mathcal G'$ is optimal with respect to the new formula for a linear combination of our optimality criteria defined as follows:
\begin{align*}
		\valp(P') = ~&\delta_1 \cdot t'_{2k} &\textit{Foremost}\\
		+ ~&\delta_2\cdot(-t'_1) &\textit{Reverse-Foremost}\\
		+ ~&\delta_3\cdot(t'_{2k} -  t'_1) &\textit{Fastest}\\
		+ ~&\delta_4\cdot \sum_{i=1}^{2k} c_\lambda(e'_i) &\textit{Shortest}\\
		+ ~&\delta_5\cdot \sum_{i=1}^{2k} c(e'_i) &\textit{Cheapest}\\
		+ ~&(\delta_6 /2 )  \cdot 2k &\textit{Minimum Hop-Count}\\
		+ ~&\delta_7 \cdot \sum_{i=1}^{2k-1} ((t'_{i+1} - (t'_i - \alp(v'_i))) \cdot \ind(v'_i)) &\textit{Minimum Waiting Time}.  
\end{align*}
\cref{transformation} runs in $O(|V| + |E|)$ time.
\end{proposition}

\begin{proof}
Now, we will show that any temporal walk~$P$ in~$G$ corresponds to a temporal walk~$P'$ in~$G'$ such that~$\val(P) = \valp(P')$ and vice versa.
To this end, observe that each vertex~$v\in V^E$ has an in-going and an out-going arc and the vertex set~$V = V'\setminus V^E$ is an independent set in~$G'$.
Hence, each temporal walk in~$G'$ from~$s$ to~$z$ alternately uses vertices in~$V$ and~$V^E$.
Since each vertex in~$V^E$ represents an arc in~$G$, each temporal walk in~$G'$ has a unique representation in~$G$ and vice versa.
It remains to show that~$\val(P) = \valp(P') - C$ for any temporal walk~$P$ in~$G$ where~$C$ is a constant only depending on the last vertex in~$P$.
Observe that by construction each vertex~$v_i$ in~$P$ is the same vertex as~$v'_{2i}$ in~$P'$ and the arc~$(v_{i-1},v_i,t_i,\lambda_i)$ is represented by the vertex~$v'_{2i}$ and the arcs~$e'_{2i-1}$ and~$e'_{2i}$.
It holds that~$\lambda_i = c_\lambda(e'_{2i}) = c_\lambda(e'_{2i-1}) + c_\lambda(e'_{2i})$ and~$c(e_i) = c'(e'_{2i}) = c'(e'_{2i-1}) + c'(e'_{2i})$.
Finally,~$t_{i+1} = t'_{2i+1}$ and~$t_i + \lambda_i = t'_{2i} - \alp(v'_{2i})$ and hence~$t_{i+1} - (t_i + \lambda_i) = t'_{2i+1} - (t'_{2i} - \alp(v'_{2i}))$.
Thus,
\begin{align*}
			\val(&P) =\\
			  ~&\delta_1 \cdot(t_k + \lambda_k) + \delta_2\cdot(-t_1) + \delta_3\cdot(t_k  + \lambda_k - t_1) + \delta_4\cdot \sum_{i=1}^k \lambda_i \\
			~&+ \delta_5\cdot \sum_{i=1}^k c(e_i) + \delta_6\cdot k + \delta_7\cdot \sum_{i=1}^{k-1} (t_{i+1} - (t_i + \lambda_i))  \\
			= ~&\delta_1 \cdot(t_k + \lambda_k + \alpha(t_k) - \alpha(t_k)) + \delta_2\cdot(-t_1) \\
			~&+ \delta_3\cdot(t_k + \lambda_k + \alpha(t_k) - \alpha(t_k) - t_1) + \delta_4\cdot \sum_{i=1}^k \lambda_i + \delta_5\cdot \sum_{i=1}^k c(e_i)\\
			~&+ \delta_6\cdot k + \delta_7\cdot \sum_{i=1}^{k-1} (t_{i+1} - (t_i + \lambda_i + \alpha(t_i) - \alpha(t_i)))  \\
			= ~&\delta_1 \cdot(t'_{2k-1} - \alp(v'_{2k})) + \delta_2\cdot(-t'_1) + \delta_3\cdot(t'_{2k-1} - \alp(v'_{2k}) - t'_1)\\
			~&+  \delta_4\cdot \sum_{i=1}^{2k} c_{\lambda}(e'_i) + \delta_5\cdot \sum_{i=1}^{2k} c'(e'_i) + (\delta_6/2)\cdot 2k\\
			~&+ \delta_7\cdot \left(\sum_{i=1}^{k} ((t'_{2i} - (t'_{2i-1} - \alp(v'_{2i-1}))) \cdot 0) + \sum_{i=1}^{k-1} (t'_{2i+1} - (t'_{2i} - \alp(v'_{2i})))\right)\\
			= ~&\delta_1 \cdot(t'_{2k-1} - \alp(v'_{2k})) + \delta_2\cdot(-t'_1) \\
			~&+ \delta_3\cdot(t'_{2k-1} - \alp(v'_{2k}) - t'_1) + \delta_4\cdot \sum_{i=1}^{2k} c_{\lambda}(e'_i) + \delta_5\cdot \sum_{i=1}^{2k} c'(e'_i)\\
			~&+ (\delta_6/2)\cdot 2k + \delta_7\cdot \sum_{i=1}^{2k-1} ((t'_{i+1} - (t'_i - \alp(v'_{i}))) \cdot \ind(v'_i))  \\
			= ~&\valp(P') - (\delta_1 + \delta_3) \cdot \alp(v'_{2k}).
\end{align*} 

Observe that~$(\delta_1 + \delta_3) \cdot \alp(v'_{2k})$ is independent of the temporal walk~$P'$ and hence any optimal walk~$P$ in~$G$ corresponds to an optimal temporal walk~$P'$ in~$G'$.
Lastly, notice that it is easy to verify that \cref{transformation} runs in $O(|V|+|E|)$ time.
%
%
\end{proof} 
The algorithm for computing optimal temporal walks that we will introduce in the forthcoming section will find temporal walks optimizing the formula $\valp(\cdot)$ introduced in~\cref{prop:transformation}.
For instantaneous temporal graphs, where we do not have to use \cref{transformation}, optimizing according to $\valp(\cdot)$ 
is not a drawback as stated in the following: 
\begin{observation}
\label{obs:instgraph}
Let \tempGTna~be an instantaneous temporal graph, let~$c\colon E \rightarrow \mathds N$ be a cost function. 
Let further $P=\left(\left (v_{i-1},v_i,t_i,\lambda_i \right)\right)_{i=1}^k=\left(e_i\right)_{i=1}^k$ 
be a temporal walk in~$\mathcal G$.
For $\delta_1,\ldots,\delta_7\in \mathds Q^+_0$, it holds that $$\val(P) = \valp(P) + (\delta_6/2) \cdot k$$ for $c' =c$, $c_\lambda\equiv 0$, $\ind \equiv 1$, and $A\equiv 0$.
\end{observation}

\section{Algorithm}
\label{sec:algorithm}
In this section, we present a single-source optimal walks algorithm with respect to any linear combination of our optimality criteria.
That is, given a temporal graph \tempGTna, a cost function~$c \colon E \rightarrow \mathds N$ and a source vertex~$s \in V$, we compute an optimal temporal walk with respect to any linear combination with $\delta_1,\ldots,\delta_7 \in \mathds Q$ from $s$ to all vertices in the temporal graph (if it exists).
To this end, we first apply \cref{transformation} to $\mathcal G$ to obtain an instantaneous temporal graph.
\cref{algo} then performs for each~$t \in [T]$ three main steps: 
\begin{description}
	\item[GraphGeneration.] Generate~$G_t$ which only contains the arcs present at time step~$t$ and add arcs from~$s$ to each vertex~$v$ in~$G_t$ that has been reached within the last~$\beta(v)$ time steps. 
	\item[ModDijkstra.] Run a modified version of Dijkstra's algorithm to compute for each~$v$ in~$G_t$ the optimal walk from~$s$ to~$v$ that arrives at time step~$t$ (if it exists).
	\item[Update.] Update a list with representations of all candidates for optimal walks (with corresponding arrival times and optimal values) from~$s$ to each~$v \in V$.
\end{description}
\smallskip

Efficiently storing and accessing the value of an optimal walk from $s$ to $v$ that arrives at a certain time step~$t$ is the heart of the algorithm. 
We can maintain this information in~$O(|E|)$ time during a run of \cref{algo} such that this information can be accessed in constant time. 
The algorithm is presented in \cref{algo} from which we can derive the following theorem: 
\begin{theorem}
With respect to any linear combination of the optimality criteria, an optimal temporal walk from a source vertex~$s$ to each vertex in a temporal graph can be computed in~$O( |V| + |E| \log |E|)$ time.
\label{theorem:algo}
\end{theorem} 
\begin{algorithm}[http!]
	\caption{Computes optimal walks.}
	\label{algo}
	\SetKwInOut{Input}{input}\SetKwInOut{Output}{output}
	\SetKwFunction{sso}{SingleSourceOptWalk}
	\SetKwFunction{gendirgraph}{generateGraph}
	\SetKwFunction{prim}{Prims}
	\SetKwFunction{dijkstra}{modDijkstra}
	\SetKwFunction{bfs}{BFS}
	\SetKwInOut{Variables}{Variables}
	\KwIn{An instantaneous temporal graph~\tempGTna, two cost functions~$c, c_\lambda$, two vertex functions $\ind,A$, and a source vertex~$s\in V$.}
	\KwOut{For each~$v \in V$ the specific length of an optimal $s$-$v$ walk.}
	\Variables{
	
	\small%
	\begin{tabularx}{.96\textwidth}{ c X }
		$\opt(v)$ 	& stores the value of an optimal walk from $s$ to $v$ within~$[0,t]$; \\
		$L(v)$ 		& is a sorted list $[(\opt_{a_1},a_1), \ldots,(\opt_{a_k},a_k)]$ where $\opt_{a_i}$ is an optimal value of a walk from $s$ to $v$ that arrives at time~$a_i$ with $t + \beta(v) \leq a_i \leq t $.\\ 
		$\delta_1,...,\delta_7$~ & linear combination of the optimality criteria \textit{foremost}, \textit{reverse-foremost}, \textit{fastest}, \textit{shortest}, \textit{cheapest}, \textit{minimum hop-count}, and \textit{minimum waiting time}, respectively.\\
 		\bottomrule
	\end{tabularx}}
%

	\BlankLine
	\SetKwProg{Fn}{function}{:}{}
		Initialize $\opt(v) = \infty$ and $L(v)$ as empty list for all $v \in V \setminus \{s\}$ \label{algoLine:in1}\\

		\For{$ t = 1,\ldots, T  ~\with~E_t \not = \emptyset$}{
			$G,d_t,d_r \leftarrow$ \gendirgraph{$G_t$} \label{algoLine:genDirG}\\
			$V',\opt_t \leftarrow$ \dijkstra{$G,d_t,d_r$} \label{algoLine:dij}\\
			\For{$v \in V' $}{ \label{algoLine:reached1}
			$\opt(v) \leftarrow \min\{\opt(v), \delta_1 \cdot t  - \delta_2 \cdot T + \delta_3 \cdot (t - T) + \opt_t(v)\}$ \label{algoLine:reached2}\\
			$L(v) \leftarrow$ append $(\opt_t(v),t)$ and delete \emph{redundant} tuples (see \cref{lemma:redundantEl})\label{algoLine:reached3} 
			}
		} 
		\Return{$\opt$} \label{algoLine:return} \\
	\Fn{\gendirgraph{$G_t$}}{\label{algoLine:gen0}
	Initialize $E_r \leftarrow \emptyset$; $d_r(v,w) \leftarrow \infty$ and $d_t(v,w) \leftarrow \infty$ for all $v,w \in V_t \cup \{s\}$\label{algoLine:gen1}\\
		\For{$(v,w) \in E_t $}{\label{algoLine:gen6}
	$d_t(v,w)\leftarrow \begin{cases} (\delta_2+ \delta_3) \cdot(T-t) + \delta_4 \cdot c_\lambda(v,w,t)+ \delta_5 \cdot c(v,w,t) + \delta_6 & \text{ if }v=s\\ 
		  								\delta_4 \cdot c_\lambda(v,w,t)+ \delta_5 \cdot c(v,w,t) + \delta_6  & \text{ else}\end{cases}$} \label{algoLine:gen7}
	\For{$v \in V_t \setminus \{s\}$}{\label{algoLine:gen2}
	delete tuples $(\opt_a,a)$ in $L(v)$ with $a + \beta(v) < t$\\
	\If{$L(v) ~\notempty$}{ \label{algoLine:gen3}
	$E_r \leftarrow E_r \cup \{(s,v)\}$ \label{algoLine:gen4}\\
	$\opt_a \leftarrow \min\{\opt_a \mid (\opt_a,a) \in L(v) \}$\\
	$d_r(s,v) \leftarrow \opt_a + \delta_7 \cdot \ind(v)\cdot(t - a + A(v)) $\label{algoLine:gen5}
	}}
	\Return{$\big((V_t \cup \{s\},E_t \cup E_r),d_t,d_r\big)$} \label{algoLine:gen8}
	}
	\Fn{\dijkstra{$(V,E_t \cup E_r),d_t,d_r$}}{\label{algoLine:dij0}
	initialize $\opt_t(v)\leftarrow\infty$, $r(v)\leftarrow\infty$ for all $v \in V_t$, and $r(s) =0$\\
	initialize $Q\leftarrow V$ and  $V'\leftarrow \emptyset$\\
	\While{$Q \not = \emptyset$}{
	$v$ $\leftarrow$ vertex in $Q$ with minimum $r(v)$\\
	remove $v$ from $Q$\\
	\For{$(v,w) \in E_t \cup E_r$}{
	$r(w) \leftarrow \min\{r(w), r(v)+ \min\{d_t(v,w),d_r(v,w)\}\}$\\
	\If{$(v,w) \in E_t$}{
	$\opt_t(w) \leftarrow \min\{\opt_t(w), r(v)+d_t(v,w)\}$\\
	$V' \leftarrow  V' \cup \{w\}$
	}
	}
	}
	\Return{$V',\opt_t$}\label{algoLine:dij1}
	}
\end{algorithm}
\paragraph{Algorithm Details.}
Let $\mathcal G'$~be a temporal graph with a cost function $c'\colon E \rightarrow \mathds N$ and let~$s \in V$ be the source. 
We can apply \cref{transformation} to $(\mathcal G',c')$ to obtain $($\tempGTna$,c,c_\lambda,\ind,A)$ where $\mathcal G$ is an instantaneous temporal graph. If $\mathcal G'$ is already an instantaneous temporal graph, then we set $c =c'$, $c_\lambda\equiv 0$, $\ind \equiv 1$, and $A\equiv 0$ as shown in~\cref{obs:instgraph}.

For~$\delta_1,\ldots,\delta_7 \in \mathds Q$, \cref{algo} computes an optimal walk~$P=\left(\left(v_{i-1},v_i,t_i\right)\right)_{i=1}^k=\left(e_i\right)_{i=1}^k$ from $s$ to $v$ for all $v\in V$ with respect to 
\begin{align*}
 		\valp(P)= &\delta_1 \cdot(t_k ) + \delta_2\cdot(-t_1) + \delta_3\cdot(t_k - t_1) + \delta_4\cdot \sum_{i=1}^k c_\lambda(e_i) \\
		&+\delta_5\cdot \sum_{i=1}^k c(e_i) + \delta_6\cdot k + \delta_7 \cdot \sum_{i=1}^{k-1} (t_{i+1} - t_i + A(v_{i+1})) \cdot \ind(v_i) .
\end{align*}
We have shown in \cref{transformation} that an optimal walk with respect to $\valp(P)$ in $\mathcal G$ directly corresponds to an optimal walk with respect to $\val(P)$ in the original temporal graph $\mathcal G'$.

For each vertex~$v \in V \setminus \{s\}$, \cref{algo} stores in 
 $\opt(v)$ the value of an optimal walk from~$s$ to~$v$ and 
in $L(v)$ a list of all relevant arrival times from~$s$ to~$v$ with their optimal values.
In the beginning, $\opt(v) = \infty$ and $L(v)$ is initially an empty list~(\cref{algoLine:in1} in \cref{algo}). 
Then, for each time step~$t$, \cref{algo} computes the optimal walk from the source~$s$ to $v \in V$ that arrives in time step~$t$ (if it exists). 
Thus, \cref{algo} performs for each $t \in \{1,\ldots, T\}$ the following steps:
\begin{description}
\item[GraphGeneration.]
Generate a static graph~$G$~with \texttt{GenerateGraph}~(\cref{algoLine:genDirG} and \crefrange{algoLine:gen0}{algoLine:gen8}). 
This graph consists of the static graph $G_t=(V_t ,E_t)$, that is, the static graph induced by all time-arcs with time stamp $t$, and the source vertex~$s$. 

The weight of an arc $(v,w) \in E_t$ is set to their arc cost~$\delta_4 \cdot c_\lambda(v,w) + \delta_5 \cdot c(v,w) + \delta_6$. 
If further $v=s$, then we have to add $(\delta_2+ \delta_3) \cdot(T-t)$ to take the departure time in~$s$ into account for the criteria \textit{reverse-foremost} and \textit{fastest}, see \cref{algoLine:gen7}.

Additionally, non-existing arcs from $s$ to each vertex~$v \in V_t$ are added if there exists a temporal walk from $s$ to~$v$ that arrived not later than~$\beta(v)$ time steps ago. 
Let $\opt_a$ be the optimal value among all walks that arrive within the time interval~$[t-\beta(v),t]$ and let $a$ be the corresponding arrival time in~$v$. 
Additionally, the minimum waiting time $A(v)$ plus the additional waiting time $(t-a)$ in $v$ has to be taken into account if $\ind(v)=1$. Hence, the weight of arc~$(s,v)$ is set to~$\opt_a + \delta_7 \cdot \ind(v) \cdot (t - a + A(v) )$, see \cref{algoLine:gen5}.
Let $E_r$ be the set of these additional arcs. 
Then,~$G = (V_t \cup \{s\}, E_t \cup E_r, d_t,d_r )$. 

\item[ModDijkstra.]
Run a modified Dijkstra Algorithm on~$G$ with \texttt{modDijkstra} (\cref{algoLine:dij} and \crefrange{algoLine:dij0}{algoLine:dij1}). 
Instead of computing a shortest walk (using the original Dijkstra Algorithm) in $G$, compute a shortest walk among all walks that end in an arc of $E_t$. 
This represents a temporal walk that arrives in time step $t$ with optimal value. 
The function \texttt{modDijkstra} returns the set $V'$ of vertices that can be reached within~$G$ via an arc in $E_t$ and the function~$\opt_t \colon V' \rightarrow \mathds N$ that maps each~vertex $v\in V'$ to its optimal value of a walk from $s$ to $v$ that arrives exactly at time $t$. 

\item[Update.]
For each $v \in V'$, set the optimum~$\opt(v)$ to the minimum of its current value and the optimal value of a newly computed walk, that is, $\opt(v) = \min\{\opt(v), \delta_1 \cdot t  + \delta_3 \cdot (t - T) + \opt_t(v)\}$~(\cref{algoLine:reached2}). Herein, we have to add the arrival time~$t$ that has not been taken into account in the calculation of the optimal value because it is the same for all walks found at time step~$t$. Add the tuple~$( \opt_t(v), t )$ to list~$L(v)$~(\cref{algoLine:reached3}).  
\end{description}
\smallskip

After the \textbf{Update} step for time step~$t$, the list $L(v)$ contains all tuples $(\opt_\text{ar},\text{ar})$ such that there exists a walk from~$s$ to~$v$ that arrives in~$\text{ar} \in [t-\beta(v), t]$ with its optimal value~$\opt_\text{ar}$. 
We want to have constant-time access to the optimal value of a walk that arrives in $v$ within time interval~$[t-\beta(v), t]$. 
This can be achieved by deleting tuples from list~$L(v)$ that are \emph{redundant}, that is these tuples are nonmeaningful for the correct computation of optimal  walks.
Let $$L(v)=[(\opt_{a_1},a_1), \ldots,(\opt_{a_k},a_k)]$$ 
be such a list for a time step $t$ with $t-\beta(v) \leq a_1 < \ldots < a_k \leq t$. 
A tuple~$(\opt_{\text{ar}},\text{ar})$ is redundant if there exists a tuple with an arrival time greater than $\text{ar}$ such that its optimal value  smaller than~$\opt_\text{ar}$ plus the additional waiting time. This is shown with the following lemma: 
\begin{lemma}
	\label{lemma:redundantEl}
	For a time step $t \in \{1,\ldots,T\}$ and a vertex $v \in V$, if there are two tuples $(\opt_{a_i},a_i),(\opt_{a_j},a_j) \in L(v)$ with $a_i < a_j$ and $$ \opt_{a_j} \leq \opt_{a_i}  + \delta_7 \cdot \ind(v) \cdot (a_j-a_i),$$ then $(\opt_{a_i},a_i)$ is redundant and can be removed from $L(v)$. 
\end{lemma} 
\begin{proof}
	After all time-arcs with time stamp~$t$ have been processed, \cref{algo} only considers time-arcs with time stamp~$t'>t$. 
	In the generated graph $G$ (\cref{algoLine:genDirG}), the algorithm adds an arc from~$s$ to~$v\in V_{t'}$ if a walk from~$s$ arrives in $v$ within~$[t'-\beta(v),t']$. 
	If~$a_i \in [t'-\beta(v),t']$, then~$a_j \in [t'-\beta(v),t']$~because~$a_i < a_j < t'$. 
	Furthermore, let~$(\opt,a)$ be the optimal value and the arrival time of a walk from $s$ to~$v$ that arrives within $[t'-\beta(v),t']$ that minimizes $$\opt + \delta_7 \cdot \ind(v) (t' - a + A(v)).$$  
	Then, the weight of the arc $(s,v)$ is set to this value. 
	Due to~$a_i < a_j \in [t'-\beta(v),t']$ and $\opt_{a_i}  + \delta_7 \cdot \ind(v) \cdot (a_j-a_i)  \geq \opt_{a_j}$, 
	we know that 
	$$\opt_{a_j} + \delta_7 \cdot \ind(v) (t' - a_j + A(v))\leq  \opt_{a_i} + \delta_7 \cdot \ind(v) (t' - a_i + A(v)).$$ 
	Hence, the tuple $(\opt_{a_i},{a_i})$ is not needed in the list~$L(v)$ at time step $t$ and can be removed.  
\end{proof} 
If $L(v)$ does not contain any redundant tuples, then it also holds that~$$\opt_{a_1} + \delta_7 \cdot \ind(v) \cdot (a_2-a_1)< \cdots < \opt_{a_k}.$$ 
Hence,~$(a_1, \opt_{a_1})$ contains the optimal value and arrival time of a walk that arrives within time interval~$[t-\beta(v), t]$ and minimizes $\opt + \delta_7 \cdot \ind(v) (t' - a + A(v)).$ 
It follows that finding~$\opt_a = \min\{\opt_a \mid (\opt_a,a) \in L(v) \}$ in \cref{algoLine:gen5} takes constant time. 
The deletion of redundant tuples takes $O(|E|)$ time during the whole run of \cref{algo}. 
With these considerations at hand, we can derive the following lemma.
\begin{lemma}
	\cref{algo} runs in $O(|V|+ |E| \log |E|)$ time.
	\label{lemma:algoFastestRunningTime}
\end{lemma}
\begin{proof}
The initialization in \cref{algo} can be done in $O(|V|)$ time. Furthermore, the time-arcs have to be sorted by time stamps which takes $O(|E| \log |E|)$ time. 
Then, for each time step~$t\in [T]$, \cref{algo} generates a static directed graph $G=(V_t \cup \{s\},E_t \cup E_r)$ with $O(|V_t|)$ vertices and $O(|E_t| + |V_t|)$ arcs which takes $O(|E_t| + |V_t|)$ time. 
 
For each generated graphs $G$, \texttt{modDijkstra} is executed in~$O(|E_t|  \log |E_t| )$ time. The updates of $\opt$ and $L$ afterwards run in $O(|V_t|)$ time. Note that $|V_t|$ is the number of vertices that have at least one in-going or out-going time-arc at time step $t$. Consequently, it holds that~$|V_t| \leq 2 |E_t|$.

Due to the sorting of~$L(v)$ as shown \cref{lemma:redundantEl}, maintaining these lists in \cref{algoLine:gen3,algoLine:reached3} takes only~$O(|E|)$ time during the whole run of the algorithm. In the list $L(v)$, we delete at most as many elements as there are time-arcs in the temporal graph. Recall that if $(\opt_a,a) \in L(v)$, then there exists a time-arc $(w,v,a) \in E$.

We can add up the running time by
\begin{align*}
&O\big(|V| + |E| + \sum_{t=1}^{T}  (|E_t|) + (|V_t| \log |V_t|) \big)\\ 
= ~&O\big(|V| + |E| + \sum_{t=1}^{T}  (|E_t|) + (|E_t| \log |E_t|) \big)\\ 
= ~&O\big(|V| + |E| + \sum_{t=1}^{T}  |E_t| \log |E_t| \big)\\ 
\subseteq ~&O\big(|V| + |E| \log |E|  \big) 
\end{align*}
Hence, \cref{algo} runs in $O\big(|V| + |E| \log |E| \big)$ time. 
\end{proof}

Next, we show the correctness of \cref{algo}.
We show that for every time step~$t$ and for every vertex~$v$, \cref{algo} computes an optimal walk from~$s$ to~$v$ that arrives at time step~$t$ (if it exists). 
\begin{lemma}  
For a time step $t \in [T]$, \cref{algo} computes the optimal value of a temporal walk from $s$ to $v \in V$ that arrives exactly in time step $t$. 
\label{lemma:algoFastestTimeStep}
\end{lemma}
\begin{proof}
The proof is by induction on the time step $t\in \{1,\ldots,T\}$.

In the beginning, $L(v)$ is empty. For $t=1$, the algorithm generates a graph $G=(V_1 \cup \{s\}, E_1)$. 
For all arcs $(s,w) \in E_1$ the weights are set to
$$d_1(s,w) = (\delta_2+ \delta_3) \cdot(T-1) + \delta_4 \cdot c_\lambda((v,w,1))+ \delta_5 \cdot c((v,w,1)) + \delta_6;$$
for the other arcs $(v,w) \in E_1$ the weights are set to $$d_1(v,w)=\delta_4 \cdot c_\lambda(v,w)+ \delta_5 \cdot c((v,w,1)) + \delta_6.$$ 
Note that if there is an optimal temporal walk arriving in time step $1$, then there also exists an optimal temporal path arriving at time step $1$.
Now if there is an optimal path~$P=\left(\left(v_{i-1},v_i,1\right)\right)^k_{i=1}=\left(e_i\right)^k_{i=1}$ from $s$ to a vertex $v \in V$ that arrives exactly in time step~$1$, 
then there exists a path~$P'=\left(\left(v_{i-1},v_i\right)\right)^k_{i=1}=\left(a_i\right)^k_{i=1}$ from~$s$ to~$v$ in~$G$ with value 
\begin{align*}
	\sum_{i=1}^k d_t(a_i) 	&= (\delta_2+ \delta_3) \cdot(T-1) + \sum_{i=1}^k \delta_4 \cdot c_\lambda(v_{i-1},v_i)+ \delta_5 \cdot c(v_{i-1},v_i) + \delta_6\\
							&= \opt_1(v)
\end{align*}
\cref{algo} finds in \texttt{modDijkstra} the path $P'$, adds $(\opt_1(v),1)$ to $L(v)$ and sets 
\begin{align*}
	\opt(v) 	= ~&\delta_1 \cdot 1 - \delta_2 \cdot T + \delta_3 \cdot (1-T) + \opt_1(v)\\
			= ~&\delta_1 \cdot 1 - \delta_2 \cdot T + \delta_3 \cdot (1-T) + (\delta_2+ \delta_3) \cdot(T-1) + \sum_{i=1}^k \delta_4 \cdot c_\lambda(v_{i-1},v_i)\\
				&+ \delta_5 \cdot c(v_{i-1},v_i) + \delta_6 \cdot k\\
			= ~&\delta_1 \cdot 1 - \delta_2\cdot 1 + \delta_3\cdot 0 + \delta_4\cdot \sum_{i=1}^k c_\lambda(e_i) \\
				&+\delta_5\cdot \sum_{i=1}^k c(e_i) + \delta_6\cdot k + \delta_7 \cdot \sum_{i=1}^{k-1} (t_{i+1} - t_i + A(v_{i})) \cdot \ind(v_i)\\
			=	~&\valp(P).	
\end{align*}
Note that~$A(v_i)=0$ for $i \in [k-1]$ by \cref{transformation} in time step $1$. 
If there exists an optimal path~$P^*$ in $G$, then this directly translates to the existence of a temporal path~$P'$ that arrives also in time step $1$ with a smaller optimal value than~$P$, 
contradicting the assumption that $P$ is optimal.

Now, let us assume that for all time steps $t' \in \{1,\ldots,t\}$ \cref{algo} computed the optimal value $\opt$ of a walk from $s$ to $v \in V$ that arrives exactly in time step~$t'$ 
and added $(\opt -\delta_1 \cdot t' + \delta_2 \cdot T - \delta_3 \cdot  (t'-T), t')$ to $L(v)$. 
If for time step $t+1$ a vertex~$v \in V$ has no in-going time-arc with time step~$t+1$, then there cannot exist a temporal walk from~$s$ to~$v$ that arrives exactly in time step~$t+1$. 
Thus, only vertices in $V_{t+1}$ are candidates for a temporal walk that arrives exactly in time step $t+1$.

Let $v \in V_{t+1}$ be a vertex such that there is a temporal walk from $s$ to $v$ that arrives exactly in time step $t+1$. 
Let~$P=\left(\left(v_{i-1},v_i,1\right)\right)^k_{i=1}=\left(e_i\right)^k_{i=1}$ be an optimal walk from~$s$ to~$v$ 
that arrives exactly in time step $t+1$ with the optimal value
\begin{align*}
	\valp(P)	
			&= \delta_1 \cdot t_1 - \delta_2\cdot t_k + \delta_3\cdot (t_k - t_1) + \delta_4\cdot \sum_{i=1}^k c_\lambda(e_i) \\
				&~~~+\delta_5\cdot \sum_{i=1}^k c(e_i) + \delta_6\cdot k + \delta_7 \cdot \sum_{i=1}^{k-1} (t_{i+1} - t_i + A(v_{i})) \cdot \ind(v_i)\\	
\end{align*}
Assume towards a contradiction that \cref{algo} does not find a walk from~$s$ to~$v$ with optimal value $\valp(P)$. 

First consider the case that $t_i = t+1$ for all $i \in [k]$, that is, all time-arcs of the temporal walk $P$ have time stamp $t+1$. Then, we can assume that $P$ is a temporal path. Hence there exists a path~$P'=\left(\left(v_{i-1},v_i\right)\right)^k_{i=1}=\left(a_i\right)^k_{i=1}$ from $s$ to $v$ in~$G_{t+1}$ and therefore in $G$ with optimal value 
\begin{align*}
	\sum_{i=1}^k d_t(a_i) 	&= (\delta_2+ \delta_3) \cdot(T-(t+1)) + \sum_{i=1}^k \delta_4 \cdot c_\lambda(v_{i-1},v_i) + \delta_5 \cdot c(v_{i-1},v_i) + \delta_6\\
							&= \opt_{t+1}(v).
\end{align*}
\cref{algo} finds in \texttt{modDijkstra} the path $P'$, adds $(\opt_{t+1}(v),t+1)$ to $L(v)$ and updates $\opt(v)$ to the minimum of $\opt(v)$ and
\begin{align*}
	 &\delta_1 \cdot (t+1) - \delta_2 \cdot T + \delta_3 \cdot (t+1-T) + \opt_{t+1}(v)\\
			= ~&\delta_1 \cdot (t+1) - \delta_2 \cdot T + \delta_3 \cdot (t+1-T) + (\delta_2+ \delta_3) \cdot(T-(t+1)) \\
				&+ \sum_{i=1}^k \delta_4 \cdot c_\lambda(v_{i-1},v_i) + \delta_5 \cdot c(v_{i-1},v_i) + \delta_6 \cdot k\\
			= ~&\delta_1 \cdot t_1 - \delta_2\cdot t_k + \delta_3\cdot (t_k - t_1) + \delta_4\cdot \sum_{i=1}^k c_\lambda(e_i) \\
				&+\delta_5\cdot \sum_{i=1}^k c(e_i) + \delta_6\cdot k + \delta_7 \cdot \sum_{i=1}^{k-1} (t_{i+1} - t_i + A(v_{i})) \cdot \ind(v_i)\\
			= ~&\valp(P)
\end{align*} 
Note again that~$A(v_i)=0$ for $i \in [k-1]$ if $t_{i-1}=t_i$ by \cref{transformation}.
Hence, we find a walk from $s$ to $v$ at time step $t+1$ with optimal value $\valp(P)$, this is a contradiction to our assumption.

Now assume for $P$ that there exists an $\ell \in \{1,\ldots,k-1\}$ such that for~$j \in [\ell]$ it holds that~$t_j < t+1$ and for~$j' \in \{i+1, \ldots,k\}$ it holds that~$t_{j'} = t+1$. 
The temporal walk~$P_\ell= \left(\left(v_{i-1},v_i,t_i\right)\right)^\ell_{i=1}=\left(e_i\right)^k_{i=1}$ is an optimal subwalk from $s$ to $v_\ell$ that arrives exactly in $t_\ell$, otherwise $P$ is not optimal because it could be improved by replacing $P_\ell$. It has an optimal value 
\begin{align*}
			\valp(P_\ell) = ~& \delta_1 \cdot t_1 - \delta_2\cdot t_\ell + \delta_3\cdot (t_\ell - t_1) + \delta_4\cdot \sum_{i=1}^k c_\lambda(e_i) \\
				& +\delta_5\cdot \sum_{i=1}^\ell c(e_i) + \delta_6\cdot k + \delta_7 \cdot \sum_{i=1}^{\ell-1} (t_{i+1} - t_i + A(v_{i})) \cdot \ind(v_i).
\end{align*}
Then 
\begin{align*}
	\opt_{t_\ell}(v_{\ell}) = ~&\opt_{P_\ell} - \delta_1 \cdot t_\ell + \delta_2 \cdot T - \delta_3 \cdot (t_\ell -T)\\
	= ~&(\delta_2+ \delta_3) \cdot(T-t_\ell) + \delta_4\cdot \sum_{i=1}^\ell c_\lambda(e_i) \\
				&+\delta_5\cdot \sum_{i=1}^\ell c(e_i) + \delta_6\cdot k + \delta_7 \cdot \sum_{i=1}^{\ell-1} (t_{i+1} - t_i + A(v_{i})) \cdot \ind(v_i).				 
\end{align*}
 
By our induction hypothesis, the tuple~$(\opt_{t_\ell}(v_{\ell}),t_\ell)$ was added to~$L(v_\ell)$. 
If the tuple~$(\opt_{t_\ell}(v_{\ell}),t_\ell)$  is not in $L(v_\ell)$ in time step $t+1$, then there must be another tuple~$(\opt_{\hat t}(v_{\ell}),\hat t)$ in $L(v_\ell)$ with $t_\ell < \hat t < t+1 < t_\ell + \beta(v_\ell) \leq \hat t + \beta(v_\ell)$ and 
	$$\opt_{t_\ell}(v_{\ell}) + \delta_7 \cdot \ind(v_\ell) (t+1 - t_\ell + A(v_\ell)) =  \opt_{\hat t}(v_{\ell}) + \delta_7 \cdot \ind(v_\ell) (t+1 - \hat t + A(v_\ell))$$ 
due to \cref{lemma:redundantEl}. Otherwise $P$ is not optimal because it could be improved by replacing~$P_\ell$ by the temporal walk represented by~$(\opt_{\hat t}(v_{\ell}),\hat t)$. 

Now consider the generated graph $G=(V_{t+1} \cup \{s\}, E_{t+1} \cup E_r)$.
The arc sequence~$P_{t+1}=\left(\left(v_{i-1},v_i\right)\right)^k_{i=\ell+1}=\left(a_i\right)^k_{i=\ell+1}$ is a path in~$G_{t+1}=(V_{t+1},E_{t+1})$
and, thus, contained in~$G$. 
The arc~$a_\ell =(s,v_{\ell})$ is contained in $E_r$ with weight $$d_r(s,v_{\ell}) = \opt_{t_\ell}(v_{\ell}) + (t+1 - t_\ell + A(v_\ell)) \cdot \ind(v_\ell).$$ 
Thus, there is a walk from $s$ to $v$ in $G$ and \texttt{modDijkstra} on $G$ returns the vertex $v$ because $a_k \in E_{t+1}$ with an optimal value
\begin{align*}
	d_r(a_\ell) + \sum_{i=\ell+1}^k d_t(a_i) 	= ~&\opt_{t_\ell}(v_i) + \ind(v_\ell)\cdot (t+1 - t_\ell + A(v_\ell)) \\
									&+ \sum_{i=\ell+1}^k \delta_4 \cdot c_\lambda(v_{i-1},v_i,t+1)+ \delta_5 \cdot c(v_{i-1},v_i,t+1) + \delta_6\\
								= ~&(\delta_2+ \delta_3) \cdot(T-t_\ell) + \sum_{i=1}^\ell \big (\delta_4\cdot c_\lambda(e_i) 
									+\delta_5\cdot c(e_i) + \delta_6  \big )\\
									&+ \delta_7 \cdot \sum_{i=1}^{\ell-1} \ind(v_i) \cdot (t_{i+1} - t_i + A(v_{i})) \\
									&+ \ind(v_\ell)\cdot (t+1 - t_\ell + A(v_\ell))\\
									&+\sum_{i=\ell+1}^k \big (\delta_4 \cdot c_\lambda(v_{i-1},v_i)+ \delta_5 \cdot c(v_{i-1},v_i) + \delta_6	\big )	\\
								= ~&(\delta_2+ \delta_3) \cdot(T-t_\ell) + \delta_4\cdot \sum_{i=1}^k c_\lambda(e_i) \\
									&+\delta_5\cdot \sum_{i=1}^k c(e_i) + \delta_6\cdot k + \delta_7 \cdot \sum_{i=1}^{\ell} (t_{i+1} - t_i + A(v_{i})) \cdot \ind(v_i)\\
								= ~&(\delta_2+ \delta_3) \cdot(T-t_\ell) + \delta_4\cdot \sum_{i=1}^k c_\lambda(e_i) \\
									&+\delta_5\cdot \sum_{i=1}^k c(e_i) + \delta_6\cdot k + \delta_7 \cdot \sum_{i=1}^{k-1} (t_{i+1} - t_i + A(v_{i})) \cdot \ind(v_i)\\
								= ~&\opt_{t+1}(v).
\end{align*}
Note that~$A(v_i)=0$ for $i \in [\ell +1,k-1]$ if $t_{i-1}=t_i$ by \cref{transformation}. 
Consequently, the tuple $(\opt_{t+1}(v),t+1)$ is added to $L(v)$ and $\opt(v)$ is set to the minimum of its current value and 
\begin{align*}
 &\delta_1 \cdot (t+1) -\delta_2 \cdot T + \delta_3 (t+1 -T) + \opt_{t+1}(v)\\
		= ~&\delta_1 \cdot (t+1) -\delta_2 \cdot T + \delta_3 (t+1 -T) \\
			&+(\delta_2+ \delta_3) \cdot(T-t_1) + \delta_4 \cdot \sum_{i=1}^k c_\lambda(e_i) \\
									&+\delta_5\cdot \sum_{i=1}^k c(e_i) + \delta_6\cdot k + \delta_7 \cdot \sum_{i=1}^{k-1} (t_{i+1} - t_i + A(v_{i})) \cdot \ind(v_i)\\
		= ~& \delta \cdot (t+1) + \delta_2 \cdot (t_k) + \delta_3 (t+1 -t_1) +\delta_4\cdot \sum_{i=1}^k c_\lambda(e_i) \\
									&+\delta_5\cdot \sum_{i=1}^k c(e_i) + \delta_6\cdot k + \delta_7 \cdot \sum_{i=1}^{k-1} (t_{i+1} - t_i + A(v_{i})) \cdot \ind(v_i)\\
									= ~&\valp(P).
\end{align*} 
This is a contradiction to our assumption.   

Lastly, observe that the algorithm only computes temporal walks that are contained in~$G$ as it only uses arcs from~$G_t$ which all correspond to temporal walks in~$G$ (single arcs or longer walks starting in~$s$).
Thus, for $t\in \{1,\ldots,T\}$, \cref{algo} computes an optimal temporal walk from~$s$ to~$v \in V$ that arrives exactly in time step $t$.    
\end{proof} 
Based on this statement, we can finally prove the correctness of \cref{algo}. 
This concludes the proof of \cref{theorem:algo}.
\begin{lemma}
\cref{algo} computes a \textit{optimal} walk from a source vertex~$s$ to all vertices.
\label{lemma:algoFastestCorrectness}
\end{lemma}  
\begin{proof}
Let $P=\left(\left(v_{i-1},v_i,1\right)\right)^k_{i=1}=\left(e_i\right)^k_{i=1}$ be a walk with minimum $\valp(P)$ among all temporal walks from $s$ to a vertex~$v$. 
The walk $P$ is also an optimal walk from $s$ to $v$ that arrives exactly in $t_k$.
This is computed by \cref{algo} in time step $t_k$ as shown in \cref{lemma:algoFastestTimeStep}. 
\end{proof}

\section{Experimental Results}
\label{sec:exp}
We implemented \cref{algo} and performed experimental studies including comparisons to existing state-of-the-art algorithms by Wu et al.~\cite{wu2016temporalpath}. 
We show that our algorithm---while being able to solve a much more general problem---can compete with these algorithms on real-world instances when computing temporal walks with no maximum-waiting-time constraints. 
We further examine the influence of different maximum-waiting-time values on the existence and structure (e.g., number of cycles) of optimal temporal walks and on the running time of \cref{algo}. 

\subsection{Setup and Statics}
We implemented \cref{algo} in C++ (v11) and performed our experiments on an Intel Xeon E5-1620 computer with~64\,GB of RAM and four cores clocked at~3.6\,GHz each.
The operating system was Debian GNU/Linux 7.0 where we we compiled the program with GCC v7.3.0 on optimization level~\mbox{-O3}.
We compare \cref{algo} to the algorithms of Wu et al.~\cite{wu2016temporalpath} using their C++ code and testing it on the same hardware and with the same compiler.
We tested our algorithm on the same freely available data sets as Wu et al.~\cite{wu2016temporalpath} from the well-established KONECT library~\cite{KONECT17}.
The graphs are listed in \cref{table:graphData} with some relevant statistics.
\begin{table}[t!]
\caption{Statistics for the real-world data sets used in our experiments (same freely available data sets as Wu et al.~\cite{wu2016temporalpath} from the KONECT library~\cite{KONECT17}). 
}
\centering
\begin{filecontents*}{test_environment/csv/graphdata.csv}
File,n,m,avg_Degree_half,avg_Degree,Timestamps,Lifetime
elec,7118,103675,14.5652,29.1304,196052,119088241
facebook-wosn-links,63731,817035,12.8201,25.6402,664446,1232231924
epinions,131828,841372,6.38235,12.7647,1878,81558001
enron,87273,1148072,13.155,26.31,438000,1401187798
digg-friends,279630,1731653,6.19266,12.38532,3043307,1247032806
ca-cit-HepPh,28093,4596803,163.628,327.256,4674,315273602
youtube-u-growth,3223585,9375374,2.90837,5.81674,406,19440001
dblp-coauthor,1282466,18241452,14.2237,28.4474,144,2398377601
flickr-growth,2302925,33140017,14.3904,28.7808,268,17017201
wikipedia-growth,1870709,39953145,21.3572,42.7144,4396,193186801
\end{filecontents*}
\pgfplotstabletypeset[columns={File,n,m,Lifetime},
	col sep=comma,
    columns/File/.style={string type,column name=File,column type = {l}},
	columns/n/.style={column name=$|V|$,precision=1,column type = {r}},
	columns/m/.style={column name=$|E|$,precision=1,column type = {r}},
	columns/Lifetime/.style={column name=$T$,precision=2,column type = {r}},
    every head row/.style ={before row=\toprule, after row=\midrule},
    every last row/.style ={after row=\bottomrule}]{test_environment/csv/graphdata.csv}
\label{table:graphData}
\end{table}%
For each optimization criterion, each $\beta \equiv c, c \in \{1,2,4,8,\ldots,2^{\lceil \log T \rceil}\}$, and each data set, \cref{algo} ran for 100 fixed source vertices of the data set chosen independently and uniformly at random to ensure comparability. 
Our open source code is freely available at \url{https://fpt.akt.tu-berlin.de/temporalwalks}.

\subsection{Findings}
In the following, we first compare \cref{algo} to the algorithm by Wu et al.~\cite{wu2016temporalpath} in terms of running times in our experiments. 
In the second part, we analyze the effect that different maximum-waiting-time values~$\beta$ have on \cref{algo}.

\subsubsection{Comparison} 
When comparing with the algorithms by Wu et al.~\cite{wu2016temporalpath}, we only use the runs with no maximum-waiting-time constraints ($\beta \equiv T$) and we tested all algorithms on the same set of randomly chosen starting vertices.
In the experiments, we could only measure a very small effect of the optimization criteria on the running time. This even holds for linear combinations. The only exception was the computation of~\textit{foremost} which was a bit faster in comparison to the computation of the other criteria.
For this reason we only include two examples here.
We chose foremost and shortest as these are the two criteria where \cref{algo} performed the best and the worst compared to the algorithms by Wu et al.~\cite{wu2016temporalpath}, respectively.
The respective findings are illustrated in the box plots in \cref{boxplots}. 
\pgfplotstableread[col sep = comma]{test_environment/earliest.txt}\sourceEarliest
\pgfplotstableread[col sep = comma]{test_environment/earliest_Wu.txt}\sourceEarliestWu
\pgfplotstableread[col sep = comma]{test_environment/shortest.txt}\sourceShortest
\pgfplotstableread[col sep = comma]{test_environment/shortest_Wu.txt}\sourceShortestWu
\begin{figure}[t]
	\begin{tikzpicture}
		\def\shiftOne{-3.6}
		\def\shiftTwo{-1.2}
		\def\shiftThree{1.2}
		\def\shiftFour{3.6}
		\begin{axis} [
				width=\hsize,
				height=0.45\hsize,
				grid,
				xtick=data,
				symbolic x coords={elec,facebook-wosn-links,epinions,enron,digg-friends,ca-cit-HepPh,youtube-u-growth,dblp-coauthor,flickr-growth,wikipedia-growth},
				ylabel={Time in seconds},
				x tick label style={rotate=40,anchor=east},
				ymode=log,
				legend columns=3,
				legend style={
					at={(0.5,1.05)},
					anchor=south
				}
			]
			\addplot[black,mark=x, only marks] table [x={filename}, y={readin_time}]{\sourceEarliestWu};
			\addlegendentry{Wu et al.'s read-in}
			\addlegendimage{black}
			\addlegendentry{Wu et al.'s \textit{foremost} }			
			\addlegendimage{blue}
			\addlegendentry{Wu et al.'s \textit{shortest} }
			
			\addplot[italyRed,mark=star, only marks] table [x={filename}, y={readin_time}]{\sourceEarliest};
			\addlegendentry{Our read-in}
			\addlegendimage{italyRed}
			\addlegendentry{Our \textit{foremost} }
			\addlegendimage{italyGreen}
			\addlegendentry{Our \textit{shortest} }

			\addplot[shift={(\shiftOne * \boxPlotWidth,0.0)},black,box plot median={runtime_median}{}] table [x={filename}] {\sourceEarliestWu};
			\addplot[shift={(\shiftOne * \boxPlotWidth,0.0)},black,box plot box={runtime_first}{runtime_third}] table [x={filename}] {\sourceEarliestWu};
			\addplot[shift={(\shiftOne * \boxPlotWidth,0.0)},black,box plot box={runtime_first}{runtime_third}] table [x={filename}] {\sourceEarliestWu};
			\addplot[shift={(\shiftOne * \boxPlotWidth,0.0)},black,box plot top whisker={runtime_max}{runtime_third}] table [x={filename}] {\sourceEarliestWu};
			\addplot[shift={(\shiftOne * \boxPlotWidth,0.0)},black,box plot bottom whisker={runtime_min}{runtime_first}] table [x={filename}] {\sourceEarliestWu};


			\addplot[shift={(\shiftTwo * \boxPlotWidth,0.0)},italyRed,box plot median={runtime_median}{}] table [x={filename}] {\sourceEarliest};					
			\addplot[shift={(\shiftTwo * \boxPlotWidth,0.0)},italyRed,box plot box={runtime_first}{runtime_third}] table [x={filename}] {\sourceEarliest};
			\addplot[shift={(\shiftTwo * \boxPlotWidth,0.0)},italyRed,box plot top whisker={runtime_max}{runtime_third}] table [x={filename}] {\sourceEarliest};
			\addplot[shift={(\shiftTwo * \boxPlotWidth,0.0)},italyRed,box plot bottom whisker={runtime_min}{runtime_first}] table [x={filename}]{\sourceEarliest};
         
		
			\addplot[shift={(\shiftThree * \boxPlotWidth,0.0)},blue,box plot median={runtime_median}{}] table [x={filename}] {\sourceShortestWu};
			\addplot[shift={(\shiftThree * \boxPlotWidth,0.0)},blue,box plot box={runtime_first}{runtime_third}] table [x={filename}] {\sourceShortestWu};
			\addplot[shift={(\shiftThree * \boxPlotWidth,0.0)},blue,box plot top whisker={runtime_max}{runtime_third}] table [x={filename}] {\sourceShortestWu};
			\addplot[shift={(\shiftThree * \boxPlotWidth,0.0)},blue,box plot bottom whisker={runtime_min}{runtime_first}] table [x={filename}] {\sourceShortestWu};


			\addplot[shift={(\shiftFour * \boxPlotWidth,0.0)},italyGreen,box plot median={runtime_median}{}] table [x={filename}] {\sourceShortest};					
			\addplot[shift={(\shiftFour * \boxPlotWidth,0.0)},italyGreen,box plot box={runtime_first}{runtime_third}] table [x={filename}] {\sourceShortest};
			\addplot[shift={(\shiftFour * \boxPlotWidth,0.0)},italyGreen,box plot top whisker={runtime_max}{runtime_third}] table [x={filename}] {\sourceShortest};
			\addplot[shift={(\shiftFour * \boxPlotWidth,0.0)},italyGreen,box plot bottom whisker={runtime_min}{runtime_first}] table [x={filename}]{\sourceShortest};

%
		\end{axis}
	\end{tikzpicture}
	\caption{%
		Running time comparison for computing \textit{foremost} and \textit{shortest} walks. For each graph there are four box plots. From left to right these correspond to the following algorithms: \textit{foremost} of Wu et al., our \textit{foremost}, \textit{shortest} of Wu et al., our \textit{shortest}.
		The boxes represent the 25\% to 75\% percentile of running times over the 100 runs for different sources on the respective temporal graph and the line within the boxes illustrates the 50\% percentile (the median). 
		The whiskers on the top and the bottom represent the best and worst running times, respectively.
		We here only depict the running times of the algorithms after the data has been read in and was preprocessed as we use \cref{transformation} to be able to cope with~$\lambda=0$.
		The two plots with the crosses show the running time of reading in the input and preprocessing it.
	}
	\label{boxplots}
\end{figure}
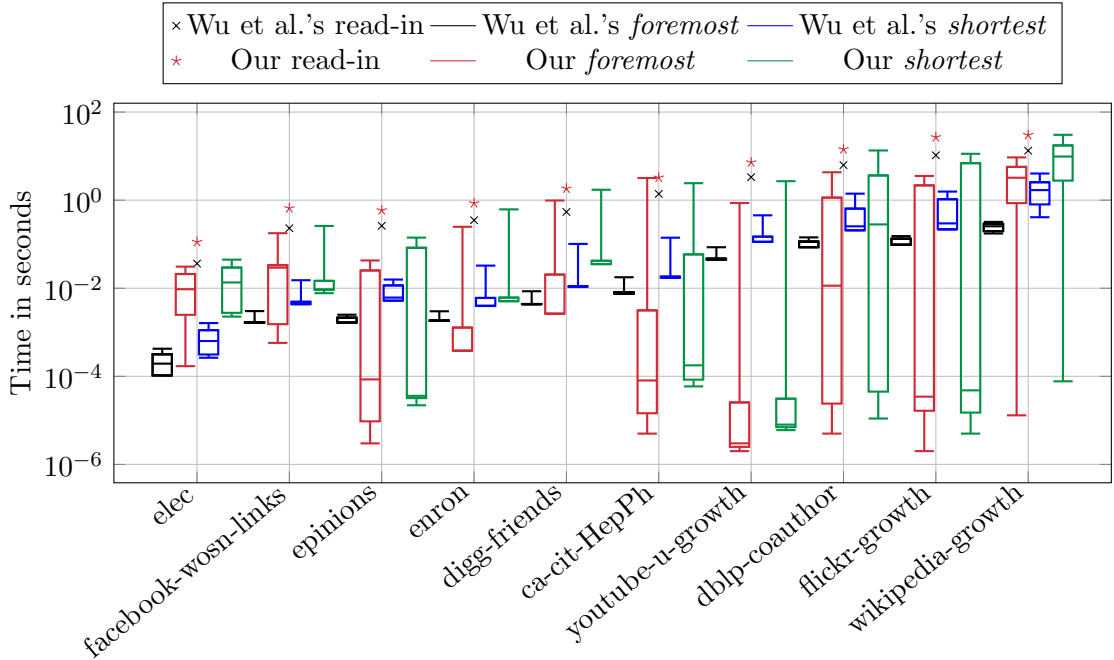%
As one can observe in \cref{boxplots}, \cref{algo} has a larger variance and is therefore more dependent on the choice of starting vertices.
This is due to the fact that \cref{algo} only considers arcs that start in vertices that were already visited while the algorithm by Wu et al.~\cite{wu2016temporalpath} always considers the whole sorted time-arc list and therefore has almost no variance in the running time.
We mention in passing that we observed that even for~$\beta\equiv T$, not all vertices can reach all other vertices by temporal walks in the considered graphs.
If one takes the running time of an average run of each algorithm, that is, the median value of running times, then both algorithms have comparable running times.
If one takes the average running time of each algorithm, then the running time of \cref{algo} is higher than the running time of the algorithm by~Wu et al.~\cite{wu2016temporalpath} by a factor of roughly ten (averaged over all optimization criteria).
Despite the fact that this is a weakness of our algorithm, we believe it to be a valuable contribution as it solves more general problems: it can easily combine multiple optimization criteria and it can cope with maximum waiting times and instantaneous arcs, that is, arcs with~$\lambda = 0$.

When looking at the time to read the data we can observe that our algorithm takes roughly twice to thrice the time for preprocessing.
This is due to the fact that for each edge in the input graph \cref{transformation} constructs a new vertex and a new edge and so the resulting graph is almost thrice the size.
The time to read in the data is much larger than the time of the actual algorithm and so \cref{algo} takes roughly thrice the time of the algorithm by~Wu et al.~\cite{wu2016temporalpath} if preprocessing is taken into account.

Finally, we compared the running time of \cref{algo} with a single optimization criterion against the same algorithm with a linear combination of all criteria considered.
\cref{boxplot:linear} displays the average and median running time for~$\beta\equiv T$ on all considered data sets.
\pgfplotstableread[col sep = comma]{test_environment/shortest_compLin.txt}\sourceShortestLin
\pgfplotstableread[col sep = comma]{test_environment/linear_combination.txt}\sourceLinear
\begin{figure}[t]
	\begin{tikzpicture}
		\begin{axis} [
				width=\hsize,
				height=0.5\hsize,
				grid,
				xtick=data,
				symbolic x coords={elec,facebook-wosn-links,epinions,enron,digg-friends,ca-cit-HepPh,youtube-u-growth,dblp-coauthor,flickr-growth,wikipedia-growth},
				ylabel={ Time in seconds},
				ylabel style={shift={(0.07,.8)}},
				legend cell align=left,
				legend pos=north west,
				x tick label style={rotate=40,anchor=east},
				y tick label style={},
				ymode=log,
				legend cell align=left,
				legend pos=north west,
			]
			\addlegendimage{}
			\addlegendimage{red}
			\addplot[black,mark=x] table [x={filename}, y={runtime_medium}]{\sourceLinear};
			\addplot[black,mark=square] table [x={filename}, y={runtime_median}]{\sourceLinear};
			\addlegendentry{Linear Combination}

			\addplot[red,mark=x] table [x={filename}, y={runtime_medium}]{\sourceShortestLin};
			\addplot[red,mark=square] table [x={filename}, y={runtime_median}]{\sourceShortestLin};
			\addlegendentry{Cheapest}

		\end{axis}
	\end{tikzpicture}
 	\caption{Running time comparison for computing optimal walks with respect to a linear combination of different optimization criteria (black) and with respect to cheapest walks (red) for~${\beta\equiv T}$. The (upper) lines with crosses illustrate the average running time and the (lower) lines with boxes show the median running time.}
 	\label{boxplot:linear}
 \end{figure}
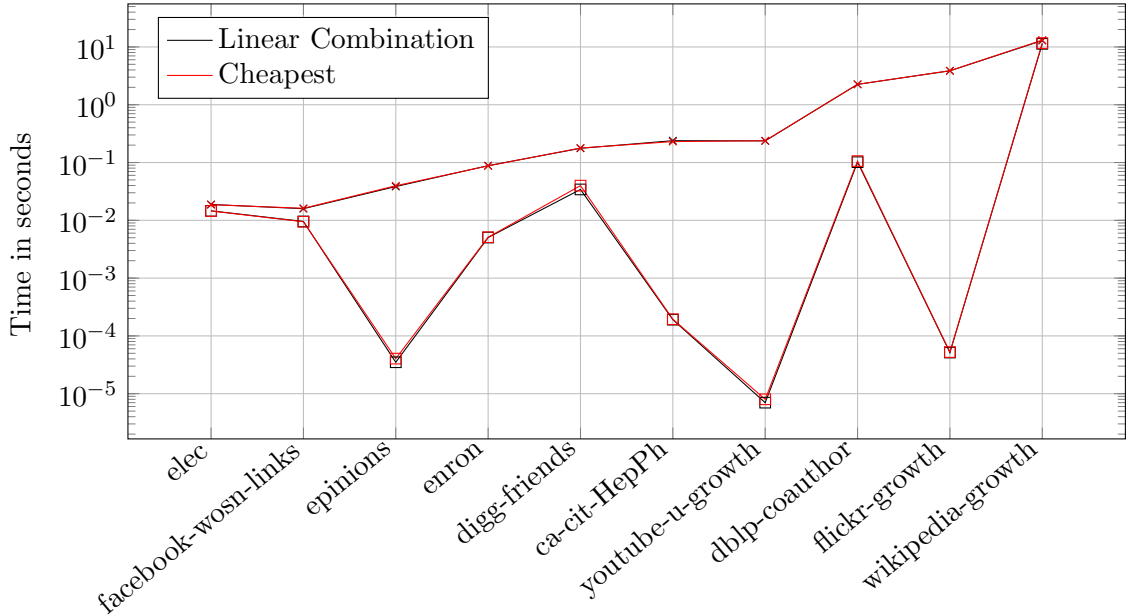%
As expected, the linear combination of optimization criteria does not change the running time compared to a single criterion.

\subsubsection{Effect of different $\beta$-values}
We next analyze the impact that~the maximum-waiting-time constraint $\beta$ has on \cref{algo}.
Decreasing~$\beta$ can have two different effects:
First, it can make temporal walks invalid as the maximum allowed waiting time in a vertex is exceeded.
Thus, with small~$\beta$-values certain vertices can only reach few vertices by temporal walks.
The second effect is that a temporal walk is invalidated but can be fixed by a detour that starts and ends in the vertex in which the maximal waiting time was exceeded.

We first investigate the second effect.
To this end, we partition the optimization criteria in two categories:
The first category contains all optimization criteria for which a detour has no negative effect on the solution.
These are \textit{foremost}, \textit{reverse-foremost}, \textit{fastest}, and \textit{minimum waiting time}.
Since the solution for, e.\,g.,\ \textit{fastest} is only depending on the first and last edge of the temporal walk, adding a cycle somewhere in between does not change the solution.
\textit{Minimum waiting time} plays a special role here as its solution can actually improve by an additional cycle.
The second category contains all other optimization criteria, that is, those for which a detour has a negative effect on the solution.
These are \textit{minimum hop count}, \textit{cheapest}, and \textit{shortest}.
Since we could not measure significant differences for the different optimization criteria within a category, we only display one figure for each category in \cref{fig:beta-impact}.\footnote{We omitted the data sets facebook-wosn-links, flickr-growth, ca-cit-HepPh, and youtube-u-growth in \cref{fig:beta-impact} to keep the figure clear. There are no additional information gains in displaying these data sets.}   

\begin{figure}[!t]
	\centering
	\begin{tikzpicture} 
		\begin{groupplot}[ 
			group style={ 
				group name=my plots, 
				group size=2 by 2, 
				xlabels at=edge bottom, 
				ylabels at=edge left, 
			}, 
			xmax=3000,
			xmin=0.00011,
			ymax=1500000,
			ymin=0.006,
			ymode=log,xmode=log,
			point meta max=5.7,
			point meta min=0.4,
			footnotesize, 
			width=0.5\hsize, 
			height=0.45\hsize, 
			ylabel={\# Cycles (Average)},
			xlabel={$\beta$-value / Lifetime in \%},
			cycle multiindex* list = {scatter src=explicit \nextlist scatter \nextlist black!40},
		] 
		\nextgroupplot[legend columns=3,
					legend style={
						at={(1.1,1.05)},
						anchor=south,
					}]
			\addplot+[mark=x, discard if not={algo}{algo1}] table[col sep=comma,y={cycles_medium}, x={beta_percent}, meta expr=lg10(\thisrow{visit_medium}) ] {test_environment/beta_values_elec.txt};
			\addlegendentry{ elec };
 			\addplot+[mark=square, discard if not={algo}{algo1}] table[col sep=comma,y={cycles_medium}, x={beta_percent}, meta expr=lg10(\thisrow{visit_medium}) ] {test_environment/beta_values_epinions.txt};
 			\addlegendentry{ epinions };
 			\addplot+[mark=triangle, discard if not={algo}{algo1}] table[col sep=comma,y={cycles_medium}, x={beta_percent}, meta expr=lg10(\thisrow{visit_medium}) ] {test_environment/beta_values_enron.txt};
 			\addlegendentry{ enron };
 			\addplot+[mark=diamond, discard if not={algo}{algo1}] table[col sep=comma,y={cycles_medium}, x={beta_percent}, meta expr=lg10(\thisrow{visit_medium}) ] {test_environment/beta_values_digg-friends.txt};
 			\addlegendentry{ digg-friends };
 			 \addplot+[mark=+, discard if not={algo}{algo1}] table[col sep=comma,y={cycles_medium}, x={beta_percent}, meta expr=lg10(\thisrow{visit_medium}) ] {test_environment/beta_values_dblp_coauthor.txt};
 			\addlegendentry{ dblp-coauthor };
 			\addplot+[mark=triangle*, discard if not={algo}{algo1}] table[col sep=comma,y={cycles_medium}, x={beta_percent}, meta expr=lg10(\thisrow{visit_medium}) ] {test_environment/beta_values_wikipedia-growth.txt};
 			\addlegendentry{ wikipedia-growth };
			\coordinate (top) at (rel axis cs:0,1);

		\nextgroupplot 
			\addplot+[mark=x,discard if not={algo}{cheapest}] table[col sep=comma,y={cycles_medium}, x={beta_percent}, meta expr=lg10(\thisrow{visit_medium}) ] {test_environment/beta_values_elec.txt};
			\addplot+[mark=square, discard if not={algo}{cheapest}] table[col sep=comma,y={cycles_medium}, x={beta_percent}, meta expr=lg10(\thisrow{visit_medium}) ] {test_environment/beta_values_epinions.txt};
 			\addplot+[mark=triangle, discard if not={algo}{cheapest}] table[col sep=comma,y={cycles_medium}, x={beta_percent}, meta expr=lg10(\thisrow{visit_medium}) ] {test_environment/beta_values_enron.txt};
 			\addplot+[mark=diamond, discard if not={algo}{cheapest}] table[col sep=comma,y={cycles_medium}, x={beta_percent}, meta expr=lg10(\thisrow{visit_medium}) ] {test_environment/beta_values_digg-friends.txt};
 			 \addplot+[mark=+, discard if not={algo}{cheapest}] table[col sep=comma,y={cycles_medium}, x={beta_percent}, meta expr=lg10(\thisrow{visit_medium}) ] {test_environment/beta_values_dblp_coauthor.txt};
 			\addplot+[mark=triangle*, discard if not={algo}{cheapest}] table[col sep=comma,y={cycles_medium}, x={beta_percent}, meta expr=lg10(\thisrow{visit_medium}) ] {test_environment/beta_values_wikipedia-growth.txt};
		
		\nextgroupplot[ylabel={Running Time (Average)},ymax=40,ymin=0.00001,title=foremost walk] 
			\addplot+[mark=x, discard if not={algo}{algo1}] table[col sep=comma,y={runtime_medium}, x={beta_percent}, meta expr=lg10(\thisrow{visit_medium}) ] {test_environment/beta_values_elec.txt};
 			\addplot+[mark=square, discard if not={algo}{algo1}] table[col sep=comma,y={runtime_medium}, x={beta_percent}, meta expr=lg10(\thisrow{visit_medium}) ] {test_environment/beta_values_epinions.txt};
 			\addplot+[mark=triangle, discard if not={algo}{algo1}] table[col sep=comma,y={runtime_medium}, x={beta_percent}, meta expr=lg10(\thisrow{visit_medium}) ] {test_environment/beta_values_enron.txt};
 			\addplot+[mark=diamond, discard if not={algo}{algo1}] table[col sep=comma,y={runtime_medium}, x={beta_percent}, meta expr=lg10(\thisrow{visit_medium}) ] {test_environment/beta_values_digg-friends.txt};
 			 \addplot+[mark=+, discard if not={algo}{algo1}] table[col sep=comma,y={runtime_medium}, x={beta_percent}, meta expr=lg10(\thisrow{visit_medium}) ] {test_environment/beta_values_dblp_coauthor.txt};
 			\addplot+[mark=triangle*, discard if not={algo}{algo1}] table[col sep=comma,y={runtime_medium}, x={beta_percent}, meta expr=lg10(\thisrow{visit_medium}) ] {test_environment/beta_values_wikipedia-growth.txt};

		\nextgroupplot[ymax=40,ymin=0.00001,title=cheapest walk]
			\addplot+[mark=x, discard if not={algo}{cheapest}] table[col sep=comma,y={runtime_medium}, x={beta_percent}, meta expr=lg10(\thisrow{visit_medium}) ] {test_environment/beta_values_elec.txt};
 			\addplot+[mark=square, discard if not={algo}{cheapest}] table[col sep=comma,y={runtime_medium}, x={beta_percent}, meta expr=lg10(\thisrow{visit_medium}) ] {test_environment/beta_values_epinions.txt};
 			\addplot+[mark=triangle, discard if not={algo}{cheapest}] table[col sep=comma,y={runtime_medium}, x={beta_percent}, meta expr=lg10(\thisrow{visit_medium}) ] {test_environment/beta_values_enron.txt};
 			\addplot+[mark=diamond, discard if not={algo}{cheapest}] table[col sep=comma,y={runtime_medium}, x={beta_percent}, meta expr=lg10(\thisrow{visit_medium}) ] {test_environment/beta_values_digg-friends.txt};
 			 \addplot+[mark=+, discard if not={algo}{cheapest}] table[col sep=comma,y={runtime_medium}, x={beta_percent}, meta expr=lg10(\thisrow{visit_medium}) ] {test_environment/beta_values_dblp_coauthor.txt};
 			\addplot+[mark=triangle*, discard if not={algo}{cheapest}] table[col sep=comma,y={runtime_medium}, x={beta_percent}, meta expr=lg10(\thisrow{visit_medium}) ] {test_environment/beta_values_wikipedia-growth.txt};
			\coordinate (bot) at (rel axis cs:1,0);
		\end{groupplot}
		
		
		\path (top|-current bounding box.north)--
                    coordinate(legendposabove)
                    (bot|-current bounding box.north);
		\begin{axis}[%
			hide axis,
			scale only axis,
			height=.001\hsize,
			width=0.85\hsize,
			at={(legendposabove.south west)},
			yshift=1.25cm,
			anchor=south,
			point meta max=5.7,
			point meta min=0.4,
			colorbar horizontal,                  
			colorbar style={
				xtick={1,2,...,5},
				scaled ticks= true,
				xticklabel pos=upper,
				xlabel={Visited Vertices (Average)},
				xlabel style={
					yshift=6em
				},
				xticklabel={ $10^{\pgfmathprintnumber{\tick}}$},
			},
			]
			\addplot [draw=none] coordinates {(0,0)};
		\end{axis}
	\end{tikzpicture} 
	\caption{Impact of different $\beta$-values on the number of cycles, running time, and on the number of vertices that can be reached by temporal walks from the chosen starting vertices.
	All plots use the same color bar.
	Two plots on left side: results for a \textit{foremost} walk. Two plots on right side: results for a \textit{cheapest} walk.
	}
	\label{fig:beta-impact}
\end{figure}
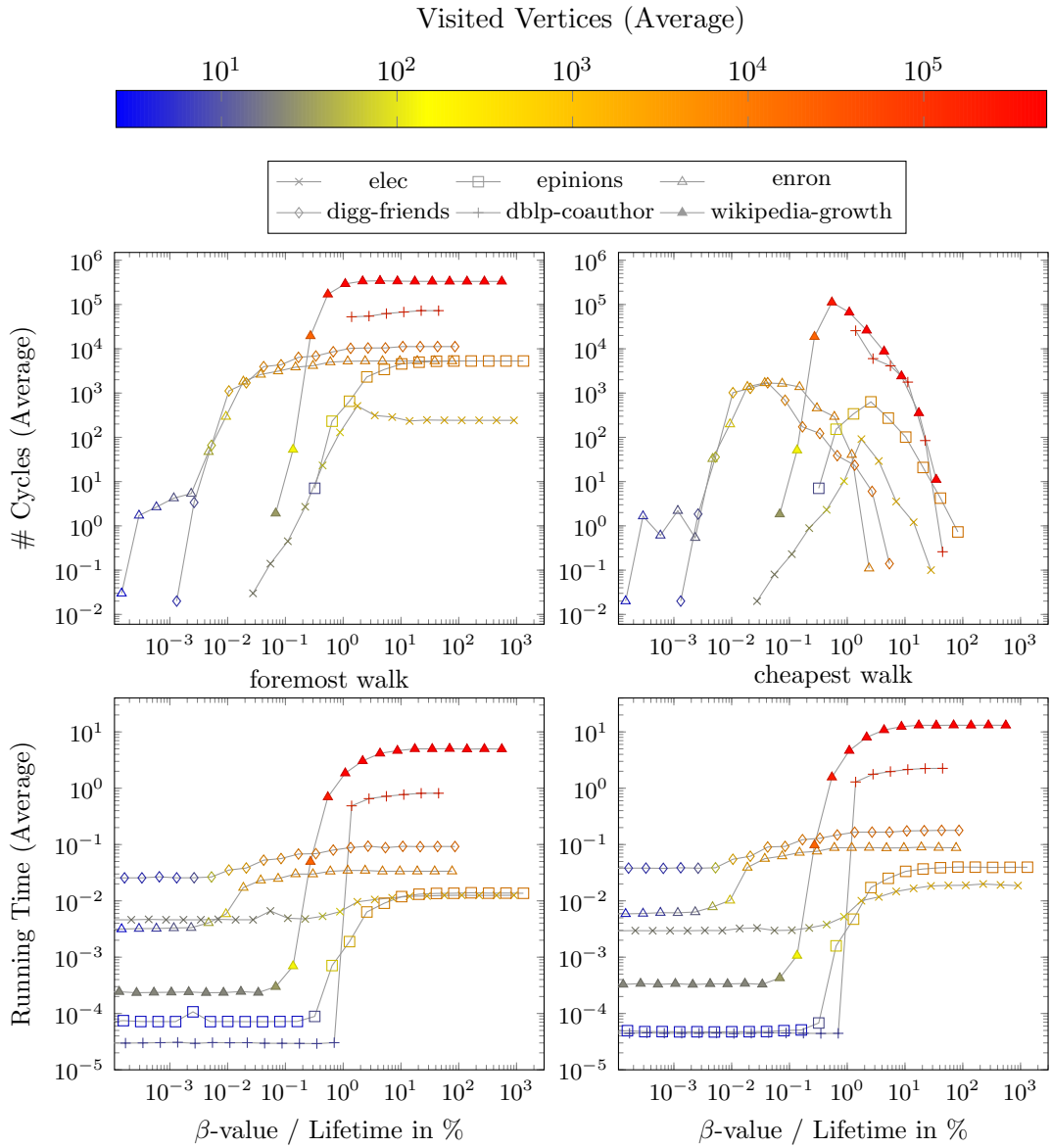

We remark that in the first category we implemented the algorithm such that cycles, which can be used but can also be omitted, are kept in the solution.
Hence \cref{fig:beta-impact} (top left plot) displays values close to the upper bound on the number of cycles in an optimal solution.

\cref{fig:beta-impact} (two bottom plots) show that the different categories behave very similarly when it comes to the running-time dependence on the value of~$\beta$. 
It seems to be more likely that the first effect we described in the beginning (that decreasing $\beta$-values can make temporal walks invalid as the maximum allowed waiting time in a vertex is exceeded) is more important for explaining the running times.
With very small~$\beta$-values, a vertex can only reach few other vertices and hence only few edges are considered by \cref{algo}.
With increasing~$\beta$-values, there seems to be a critical value (around $0.1\%-10\%$ of the lifetime of the temporal graph) where suddenly much more connections appear and hence the running time increases drastically.
This observation is affirmed by \cref{fig:connectedVertices}, which shows that (almost) independently of the input graph, the running time is linearly depending on the number of vertices that are visited. 
 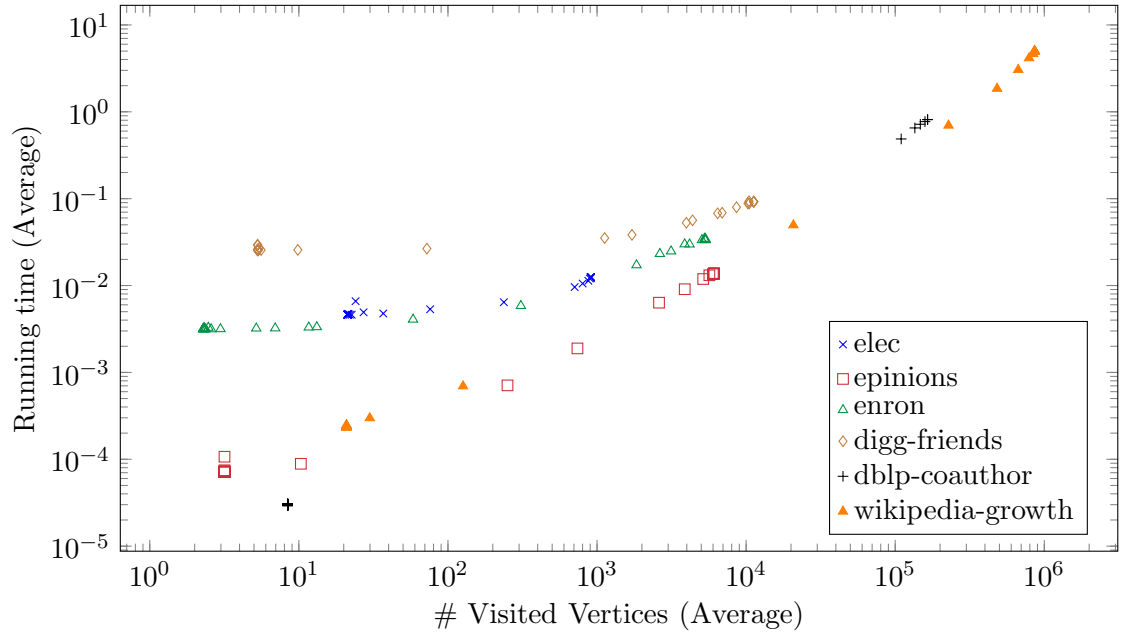
\begin{figure}[!t]
	\begin{tikzpicture}
		\begin{loglogaxis}[
				width=\hsize,
				height=0.6\hsize,
				ylabel={Running time (Average)},
				xlabel={\# Visited Vertices (Average)},
				legend cell align=left,
				legend pos=south east,
				cycle list = {blue,italyRed,italyGreen,brown,black,orange},
			]
			\addplot+[only marks,mark=x, discard if not={algo}{algo1}] table[col sep=comma,y={runtime_medium}, x={visit_medium}] {test_environment/beta_values_elec.txt}; 
			\addlegendentry{elec};
 			\addplot+[only marks,mark=square, discard if not={algo}{algo1}] table[col sep=comma,y={runtime_medium}, x={visit_medium}] {test_environment/beta_values_epinions.txt};
 			\addlegendentry{epinions};
 			\addplot+[only marks,mark=triangle, discard if not={algo}{algo1}] table[col sep=comma,y={runtime_medium}, x={visit_medium}] {test_environment/beta_values_enron.txt};
 			\addlegendentry{enron};
 			\addplot+[only marks,mark=diamond, discard if not={algo}{algo1}] table[col sep=comma,y={runtime_medium}, x={visit_medium}] {test_environment/beta_values_digg-friends.txt};
 			\addlegendentry{digg-friends};
			\addplot+[only marks,mark=+, discard if not={algo}{algo1}] table[col sep=comma,y={runtime_medium}, x={visit_medium}] {test_environment/beta_values_dblp_coauthor.txt};
 			\addlegendentry{dblp-coauthor};
 			\addplot+[only marks,mark=triangle*, discard if not={algo}{algo1}] table[col sep=comma,y={runtime_medium}, x={visit_medium}] {test_environment/beta_values_wikipedia-growth.txt};
 			\addlegendentry{wikipedia-growth};
		\end{loglogaxis}
	\end{tikzpicture}
	\caption{Average number of visited vertices and its influence on the average running time for \textit{foremost}.}
	\label{fig:connectedVertices}
\end{figure}%
We believe that the difference for small~$\beta$-values comes from the initialization which is again more depending on the input graph.
This would also confirm our explanation why our algorithm has a higher variance in running time compared to the algorithm by~\cite{wu2016temporalpath}.

\section{Conclusion}
\label{sec:conclusion}
Building on and widening previous work of Wu et al.~\cite{wu2016temporalpath}, we provided a 
theoretical and experimental study of computing optimal temporal walks under 
waiting-time constraints. 
The performed experiments indicate the practical relevance of our approach. 
As to future challenges, recall that moving from walks to paths 
would yield NP-hard optimization problems~\cite{casteigts}. Hence, for the path scenario 
the study of approximation, fixed-parameter, or heuristic algorithms is a natural next step.  
For the scenario considered in this work, note that we did not study the natural extension to 
Pareto-optimal walks (under several optimization criteria).
Moreover, for (temporal) network centrality measures based on shortest paths and walks,
counting or even listing \emph{all} temporal walks or paths would be of interest.

\paragraph*{Acknowledgment}
We thank Fabian Jacobs for his programming work helping to enable our experimental studies
and the anonymous reviewers of \emph{COMPLEX NETWORKS 2019} for their constructive feedback.
\bibliographystyle{abbrvnat}
\bibliography{literature}

\end{document}